\documentclass[11pt]{article}

\usepackage{amsmath,amssymb,amsthm,mathrsfs}
\usepackage[a4paper,margin=30mm]{geometry}
\usepackage[hidelinks]{hyperref}

\numberwithin{equation}{section}
\newtheorem{theorem}{Theorem}[section]
\newtheorem{lemma}[theorem]{Lemma}
\newtheorem{proposition}[theorem]{Proposition}
\newtheorem{corollary}[theorem]{Corollary}
\newtheorem{conjecture}[theorem]{Conjecture}
\theoremstyle{definition}
\newtheorem{definition}[theorem]{Definition}
\newcommand{\Vir}{\mathrm{Vir}}
\newcommand{\ket}[1]{\lvert#1\rangle}
\newcommand{\vac}{\lvert0\rangle}

\title{A Combinatorial Formula and a Hilbert Scheme Realization of a Rank-One Virasoro Irregular Conformal Block}
\author{Hajime Nagoya\\[0.4em]
  \small Faculty of Global Standard Education\\
  \small Institute of Liberal Arts and Science, Kanazawa University\\
  \small\href{mailto:nagoya@se.kanazawa-u.ac.jp}{\texttt{nagoya@se.kanazawa-u.ac.jp}}}
\date{}
\hypersetup{
  pdftitle={A Combinatorial Formula and a Hilbert Scheme Realization of a Rank-One Virasoro Irregular Conformal Block},
  pdfauthor={Hajime Nagoya}
}

\begin{document}

\maketitle

\begin{abstract}
For general central charge, we give a combinatorial formula and a geometric realization of the formal expansion of the rank-one Virasoro irregular conformal block of type $(0,0,1)$ at its irregular singular point. We realize the normalized block as a coefficientwise matrix element on the localized equivariant Hilbert--Fock space associated with Hilbert schemes of points on $\mathbb C^2$. At $c=1$, we prove the author's conjecture and express the parameter-independent coefficients in terms of skew-tableau numbers and Littlewood--Richardson coefficients, establishing their nonnegativity and integrality. The proof combines sectorial asymptotics of screening integrals with Laguerre--Jack evaluations. We express the screening kernel using the flow of a degree-raising differential operator and compute the Jack coefficients of its exponential from the one-box Pieri rule and the shifted Jack coefficient formula of Okounkov and Olshanski.
\end{abstract}

\section{Introduction}
\label{sec:introduction}

The AGT correspondence \cite{AldayGaiottoTachikawa2010} relates regular Virasoro conformal blocks to Nekrasov instanton partition functions, which are generating functions of equivariant integrals over instanton moduli spaces. Motivated by this geometric description, we seek geometric realizations of the expansions of irregular conformal blocks at irregular singular points. For the block of type $(0,0,1)$, we realize the normalized expansion as a coefficientwise matrix element on the localized equivariant Hilbert--Fock space associated with Hilbert schemes of points on $\mathbb C^2$. The operators and states are constructed from incidence correspondences, Hopf operations, and tautological Euler classes. The fixed-point expansion of this matrix element gives an explicit combinatorial formula.

This three-point block has two regular singular points and one irregular singular point of rank one. In terms of the irregular vertex operators $\Phi^\Delta_{\Lambda',\Lambda}(z)$ introduced in \cite{Nagoya2015}, it has the form
\[
 \mathcal B_{001}(z)
 =\langle\Delta_\infty|
 \Phi^\Delta_{\Lambda',\Lambda}(z)|\Lambda\rangle
 =z^\kappa e^{\beta/z}\left(1+\sum_{k\ge1}b_kz^k\right),
\]
expanded at the irregular singular point $z=0$. Here $|\Lambda\rangle$ is an irregular vector, $\Phi^\Delta_{\Lambda',\Lambda}(z)$ is a formal vertex operator between rank-one irregular Verma modules, and $\langle\Delta_\infty|$ is a highest-weight functional. By the existence and uniqueness theorem in \cite{Nagoya2015}, the coefficients $b_k$ exist and are uniquely determined by the Virasoro relations. This block also appears in the Zak-series representations of quantum Painlev\'e~V tau functions proposed by Bonelli, Shchechkin, and Tanzini \cite[Sections~4.1 and~6.2.1]{BonelliShchechkinTanzini2026}. Irregular conformal blocks have also been used to study partition functions of Argyres--Douglas theories \cite{GaiottoTeschner2012,NishinakaUetoko2019}.

For this block at central charge $c=1$, the author conjectured a combinatorial formula for the series, with dependence on the conformal parameters carried by explicit products over Young diagrams \cite{Nagoya2016} (see Conjecture~\ref{conj:c-one-combinatorial-formula} in Section~\ref{subsec:combinatorial-conjecture}). The remaining, parameter-independent coefficients were conjectured to be nonnegative integers, indexed by two partitions and two subdiagrams of equal size. We prove this conjecture, extend the formula to general central charge, and give a geometric interpretation on Hilbert schemes of points using Nakajima's construction \cite{Nakajima1999,NakajimaMoreLectures}.

The starting point of the proof is a screening integral, a multiple integral arising from a free-field realization of the Virasoro algebra. Asymptotic expansions of screening integrals over contours with a moving endpoint were studied by the author and Nakagawa in \cite{NagoyaNakagawa2026}. Here we treat two sets of screening variables, with arbitrary multiplicities $m,n\in\mathbb Z_{\ge0}$, initially working on the corresponding specialization of the conformal parameters. Under suitable convergence conditions, we establish the asymptotic expansion of the normalized screening integral and identify it with $\Phi^\Delta_{\Lambda',\Lambda}(z)|\Lambda\rangle$ (Theorem~\ref{thm:screening-vector-identification}). Taking the normalized pairing then gives a scalar integral whose formal asymptotic expansion is $\mathcal B_{001}(z)$ (Corollary~\ref{cor:screening-block-identification}).

Jack expansions of screening kernels, combined with Selberg-type evaluations, yield combinatorial expressions for the regular conformal blocks and their confluent limits \cite{ItoyamaOota2010,ItoyamaOotaYonezawa2010}. The representation-theoretic derivation of the AGT formula also involves Jack polynomials and Selberg integrals \cite{AlbaFateevLitvinovTarnopolskiy2011}. A Schur expansion of the present screening kernel at $c=1$ was considered in \cite{NagoyaNumata2017}.

We express the screening kernel as a differential-operator flow. The $z$-dependent factor in the scalar screening integral is
\[
 G(z;x,y)=\prod_{i=1}^{m}(1-zx_i)^{\gamma_+\gamma_\infty}\prod_{j=1}^{n}(1-zy_j)^{\gamma_+\gamma_z}\prod_{i=1}^{m}\prod_{j=1}^{n}\bigl(1-z(x_i+y_j)\bigr)^{\gamma_+^2}.
\]
Here $\gamma_+$ is the screening charge, parametrizing the central charge and determining the Jack parameter $\alpha=2/\gamma_+^2$, while $\gamma_z$ and $\gamma_\infty$ are the external charges.
The interaction factors in $G(z;x,y)$ contain $1-z(x_i+y_j)$, and the elementary identity
\[
 1-z(x_i+y_j)=(1-zx_i)(1-zy_j)(1-X_iY_j),
 \qquad X_i=\frac{zx_i}{1-zx_i},\quad
 Y_j=\frac{zy_j}{1-zy_j},
\]
reduces the interaction to a multiplicative kernel in the transformed variables. To implement this substitution, work in the algebra of symmetric functions generated by the power sums $p_r=\sum_i x_i^r$, with $\deg p_r=r$, and put
\begin{equation*}
 \mathsf D_a=\sum_{r\ge1}r p_{r+1}\frac{\partial}{\partial p_r}+ap_1,
 \qquad
 \mathsf T_a(z)=e^{z\mathsf D_a}z^{\mathsf E},
\end{equation*}
where $\mathsf E$ acts by $d$ on homogeneous degree $d$. Since $\mathsf D_a$ raises degree by one, its exponential is defined degree by degree on the completion of this algebra. Its derivation part generates the change $x_i\mapsto x_i/(1-zx_i)$, while the multiplication term supplies the accompanying weight. For finitely many variables, including the degree scaling $z^{\mathsf E}$ gives
\begin{equation}
 \mathsf T_a(z)f(x)
 =\prod_i(1-zx_i)^{-a}
 f\!\left(\frac{zx_1}{1-zx_1},\frac{zx_2}{1-zx_2},\ldots\right).
 \label{eq:intro-Fock-action}
\end{equation}

The second role of the exponential is to give explicit coefficients in the Jack basis. For generic $\alpha$, the monic Jack polynomials $P_\lambda^{(\alpha)}$, indexed by partitions $\lambda$, form a basis of the algebra of symmetric functions. Put $\theta=\alpha^{-1}$, and let $[P_\lambda^{(\alpha)}]$ denote extraction of the coefficient of $P_\lambda^{(\alpha)}$. Theorem~\ref{thm:exponential-Jack-coefficients} gives
\begin{equation}
 [P_\lambda^{(\alpha)}]e^{\mathsf D_a}P_\nu^{(\alpha)}
 =A_{\lambda,\nu}^{(\alpha)}(a)
 =\frac{P_\nu^*(\lambda;\theta)}{H_\theta(\lambda)}
 (a)_{\lambda/\nu}^{(\theta)},\qquad \nu\subseteq\lambda.
 \label{eq:intro-A}
\end{equation}
Here $P_\nu^*$ is the shifted Jack polynomial of Okounkov and Olshanski
\cite{OkounkovOlshanski1997}, normalized by
$P_\nu^*(\nu;\theta)=H_\theta(\nu)$. The factor $H_\theta(\lambda)$ is a Jack hook product, and $(a)_{\lambda/\nu}^{(\theta)}$ is the product of $a+j-1-\theta(i-1)$ over the boxes $(i,j)$ in $\lambda\setminus\nu$; see \eqref{eq:hook-products} and \eqref{eq:skew-Pochhammer}. We derive \eqref{eq:intro-A} from the one-box Pieri rule and the shifted Jack coefficient formula of Okounkov and Olshanski \cite{OkounkovOlshanski1997}. Together with \eqref{eq:intro-Fock-action}, formula \eqref{eq:intro-A} recovers the fractional-linear Jack identity of Lascoux and Warnaar \cite{LascouxWarnaar2011}.

The mixed kernel has the Jack expansion
\[
 K_w(x,y)
 =\prod_{i,j}(1-x_i y_j)^{-w}
 =\sum_{|\nu|=|\eta|}K_{\nu,\eta}^{(\alpha)}(w)
 P_\nu^{(\alpha)}(x)P_\eta^{(\alpha)}(y).
\]
Formula~\eqref{eq:K-definition} expresses the mixed-kernel coefficients as finite weighted sums of products of power-sum-to-Jack transition coefficients. The equal-size condition holds because each bihomogeneous component has the same degree in $x$ and $y$. At the screening specialization $w=-\gamma_+^2$, combining this expansion with \eqref{eq:intro-Fock-action} and \eqref{eq:intro-A} gives the coefficient $z^{|\lambda|+|\mu|}\mathcal C_{\lambda,\mu}$ of $P_\lambda^{(\alpha)}(x)P_\mu^{(\alpha)}(y)$ in $G(z;x,y)$, where
\begin{equation*}
 \mathcal C_{\lambda,\mu}
 =\sum_{\substack{\nu\subseteq\lambda,\ \eta\subseteq\mu\\|\nu|=|\eta|}}
 A_{\lambda,\nu}^{(\alpha)}(a_{\mathrm L})
 K_{\nu,\eta}^{(\alpha)}(-\gamma_+^2)
 A_{\mu,\eta}^{(\alpha)}(a_{\mathrm R}).
\end{equation*}
This explains the two subdiagrams in the conjectured expansion: the mixed kernel supplies $\nu$ and $\eta$, and the two exponentials enlarge them to $\lambda$ and $\mu$.

The two Laguerre--Jack evaluations give our main formula, stated in Theorem~\ref{thm:screening-Jack-evaluation}. For generic conformal parameters,
\begin{equation}
 \begin{aligned}
 \mathcal B_{001}(z)
 ={}&z^\kappa e^{\beta/z}
 \sum_{\lambda,\mu}z^{|\lambda|+|\mu|}\mathcal C_{\lambda,\mu}\\
 &\times (-\gamma_+c_1)^{-|\lambda|}
 \frac{(x_-)_\lambda^{(\theta)}(x_+)_\lambda^{(\theta)}}{H'_\theta(\lambda)}
 \ (\gamma_+c_1)^{-|\mu|}
 \frac{(y_-)_\mu^{(\theta)}(y_+)_\mu^{(\theta)}}{H'_\theta(\mu)},
 \end{aligned}
 \label{eq:intro-main-formula}
\end{equation}
Here $|\lambda|$ is the number of boxes of $\lambda$, $(a)_\lambda^{(\theta)}$ denotes the content product for $\nu=\varnothing$, and $H'_\theta$ is the second Jack hook product in \eqref{eq:hook-products}. The parameters $x_\pm,y_\pm,a_{\mathrm L},a_{\mathrm R}$ are explicit functions of the conformal data, listed in \eqref{eq:generic-prefactor-parameters}--\eqref{eq:generic-core-parameters}. The coefficient $b_k$ is thus a finite sum over pairs with $|\lambda|+|\mu|=k$.

The identity \eqref{eq:intro-main-formula}, first established for integer screening multiplicities, follows for general conformal parameters by polynomiality of the coefficients. At $c=1$, one has $\alpha=1$ and the Jack polynomials become Schur functions; this specialization yields the combinatorial formula in Conjecture~\ref{conj:c-one-combinatorial-formula}. For
\[
 \nu\subseteq\lambda,\qquad \eta\subseteq\mu,\qquad |\nu|=|\eta|,
\]
the parameter-independent coefficients $c_{\lambda,\mu}^{\nu,\eta}$ are given by
\[
 c_{\lambda,\mu}^{\nu,\eta}
 =P_\nu^*(\lambda;1)\,P_\eta^*(\mu;1)
 \sum_{|a|+|b|=|\nu|}
 \mathrm{LR}_{a,b}^{\nu}
 \mathrm{LR}_{a',b'}^{\eta},
\]
where $a'$ and $b'$ denote the conjugate partitions of $a$ and $b$, respectively.
This expression in terms of shifted Schur values and Littlewood--Richardson coefficients establishes their nonnegativity and integrality.

The operator description has a geometric counterpart on Hilbert schemes. The scheme $\operatorname{Hilb}^N(\mathbb C^2)$ parametrizes length-$N$ subschemes of the plane, and its torus-fixed points are indexed by partitions of $N$. After equivariant localization, Nakajima's correspondence identifies suitably normalized fixed-point classes with the Jack basis.

On the associated Hilbert--Fock space, incidence correspondences realize $\mathsf D_a$, while the Hopf algebra structure yields the mixed-kernel operator $\mathsf K(w)$, acting by the substitution $p_r\mapsto\alpha w p_r$. The normalized Laguerre evaluation functionals are represented by completed states $\mathscr L_x,\mathscr L_y$ constructed from tautological Euler classes. Theorem~\ref{thm:rank-one-Hilbert-Fock} gives
\begin{equation*}
 \mathcal B_{001}(z)
 =z^\kappa e^{\beta/z}
 \left\langle\mathscr L_x\middle|
 \mathsf T_{a_{\mathrm L}}(z)\mathsf K(-\gamma_+^2)
 \mathsf T_{a_{\mathrm R}}(z)^\dagger
 \middle|\mathscr L_y\right\rangle_T.
\end{equation*}
Here $\dagger$ is the bilinear adjoint for the localized equivariant pairing, and the states $\mathscr L_x,\mathscr L_y$ are specified in \eqref{eq:generic-Laguerre-caps}. The matrix element is interpreted coefficientwise: at each power of $z$, only finitely many components of the two states contribute.

The paper is organized as follows. Section~\ref{sec:virasoro-irregular-conformal-block} recalls irregular vertex operators and the rank-one irregular block, and states the conjectured combinatorial formula at $c=1$. Section~\ref{sec:screening-integral} constructs the screening integral, proves the required sectorial estimates, and identifies the formal expansion of its normalized scalar pairing with the irregular block. Section~\ref{sec:jack-expansion} derives the explicit Jack coefficients of $e^{\mathsf D_a}$ and combines the kernel flow with Laguerre--Jack evaluations to obtain the combinatorial formula for general central charge and prove the $c=1$ conjecture. Section~\ref{sec:hilbert} gives geometric realizations of the operators and Laguerre evaluation functionals on Hilbert schemes, leading to a matrix-element expression for the normalized block.

\section{Virasoro irregular conformal block}
\label{sec:virasoro-irregular-conformal-block}

\subsection{Irregular vertex operators}
\label{subsec:irregular-vertex-operators}

Fix a central charge $c\in\mathbb C$. The Virasoro algebra is the complex
Lie algebra
\[
 \Vir=\bigoplus_{n\in\mathbb Z}\mathbb C L_n\oplus\mathbb C C
\]
with commutation relations
\begin{equation}
 [L_m,L_n]=(m-n)L_{m+n}
   +\frac{m^3-m}{12}\delta_{m+n,0}C,
 \qquad [C,L_n]=0.
 \label{eq:virasoro-relations}
\end{equation}
Throughout this paper, $C$ acts as multiplication by $c$.

For $r\in\mathbb Z_{\ge0}$, set
\[
 \Vir_{\ge r}=\mathbb C C\oplus\bigoplus_{n\ge r}\mathbb C L_n.
\]
Given $\Lambda=(\Lambda_r,\ldots,\Lambda_{2r})\in\mathbb C^{r+1}$,
let $\mathbb C|\Lambda\rangle$ be the one-dimensional
$\Vir_{\ge r}$-module on which
\begin{equation*}
 C|\Lambda\rangle=c|\Lambda\rangle,\qquad
 L_n|\Lambda\rangle=
 \begin{cases}
  \Lambda_n|\Lambda\rangle,&r\le n\le2r,\\
  0,&n>2r.
 \end{cases}
\end{equation*}
Here and below, $U(\mathfrak g)$ denotes the universal enveloping algebra
of a Lie algebra $\mathfrak g$.

\begin{definition}
The irregular Verma module of rank $r$ with parameters $\Lambda$ is
\begin{equation*}
 M_\Lambda^{[r]}
 =U(\Vir)\otimes_{U(\Vir_{\ge r})}\mathbb C|\Lambda\rangle.
\end{equation*}
\end{definition}

We call $|\Lambda\rangle$ the rank-$r$ irregular vector.
It generates $M_\Lambda^{[r]}$ as a Virasoro module.

By the Poincar\'e--Birkhoff--Witt theorem, a basis of $M_\Lambda^{[r]}$
is given by
\[
 L_{i_1}\cdots L_{i_k}|\Lambda\rangle,
 \qquad i_1\le\cdots\le i_k<r,\quad k\ge0.
\]
For $r=0$, $M_\Lambda^{[0]}$ is the ordinary Verma module of highest weight
$\Lambda_0$, denoted by $M_{\Lambda_0}$.
For $r\ge1$, we assume $\Lambda_{2r}\ne0$ whenever invoking the
existence and uniqueness theorem below.

\begin{definition}
Let $r\ge1$ and $\Delta\in\mathbb C$. A rank-zero irregular vertex
operator of conformal weight $\Delta$ from $M_\Lambda^{[r]}$ to
$M_{\Lambda'}^{[r]}$ is a formal linear operator
$\Phi_{\Lambda',\Lambda}^{\Delta}(z)$ satisfying
\begin{equation}
 L_n\Phi_{\Lambda',\Lambda}^{\Delta}(z)
 -\Phi_{\Lambda',\Lambda}^{\Delta}(z)L_n
 =z^n\left(z\frac{d}{dz}+(n+1)\Delta\right)
       \Phi_{\Lambda',\Lambda}^{\Delta}(z),
 \qquad n\in\mathbb Z,
 \label{eq:irregular-vertex-commutator}
\end{equation}
and the normalization
\begin{equation}
 \Phi_{\Lambda',\Lambda}^{\Delta}(z)|\Lambda\rangle
 =z^\kappa\exp\left(\sum_{j=1}^{r}\beta_jz^{-j}\right)
     \sum_{m=0}^{\infty}v_mz^m,
 \qquad v_0=|\Lambda'\rangle,
 \quad v_m\in M_{\Lambda'}^{[r]},
 \label{eq:irregular-vertex-expansion}
\end{equation}
for some $\kappa,\beta_1,\ldots,\beta_r\in\mathbb C$.
\end{definition}

The power $z^\kappa$ and the exponential in
\eqref{eq:irregular-vertex-expansion} are treated as formal factors.
On an arbitrary vector of the source module, the operator takes values in
\[
 z^\kappa\exp\left(\sum_{j=1}^{r}\beta_jz^{-j}\right)
 M_{\Lambda'}^{[r]}((z)),
\]
where $M((z))$ denotes the space of formal Laurent series with coefficients
in $M$ and only finitely many negative powers.
The term rank zero refers to the vertex operator; both modules have rank $r$.

For a partition $\lambda=(\lambda_1,\ldots,\lambda_k)$, set
\[
 L_\lambda^{[r]}|\Lambda'\rangle
 =L_{r-\lambda_1}\cdots L_{r-\lambda_k}|\Lambda'\rangle,
 \qquad |\lambda|=\sum_{j=1}^{k}\lambda_j.
\]
These vectors form a PBW basis of $M_{\Lambda'}^{[r]}$. Assign degree
$|\lambda|$ to $L_\lambda^{[r]}|\Lambda'\rangle$, with the empty
partition giving $|\Lambda'\rangle$ of degree zero. The degree of a
nonzero vector is the largest degree of a PBW basis vector occurring
with nonzero coefficient; we put $\deg(0)=-\infty$. For $m\ge0$, define
\begin{equation*}
 U_m=\{v\in M_{\Lambda'}^{[r]}\mid\deg(v)\le m\}
 =\operatorname{span}_{\mathbb C}
   \{L_\lambda^{[r]}|\Lambda'\rangle\mid|\lambda|\le m\}.
\end{equation*}
We also introduce the coefficient ring
\begin{equation*}
 \mathcal R
 =\mathbb C[c,\Delta,\beta_r,\Lambda_r,\ldots,\Lambda_{2r},
            \Lambda_{2r}^{-1}],
\end{equation*}
where the parameters are regarded as indeterminates before specialization.

\begin{theorem}[Nagoya, {\cite[Theorem~2.12]{Nagoya2015}}]
\label{thm:irregular-vertex-existence}
Let $r\ge1$ and $\Lambda_{2r}\ne0$. For fixed $c$, $\Lambda$,
$\Delta$, and $\beta_r$, there exists a unique normalized rank-zero
irregular vertex operator
$\Phi_{\Lambda',\Lambda}^{\Delta}(z)$.
Its target parameters satisfy
\begin{equation*}
 \Lambda'_n=\Lambda_n-\delta_{n,r}\,r\beta_r,
 \qquad r\le n\le2r.
\end{equation*}
These data also determine $\kappa$, $\beta_1,\ldots,\beta_{r-1}$,
and every coefficient $v_m$ uniquely.
Moreover, $\kappa,\beta_1,\ldots,\beta_{r-1}$ belong to $\mathcal R$,
and for every $m\ge0$ we have $v_m\in U_m$. In the PBW expansion
\begin{equation*}
 v_m=\sum_{|\lambda|\le m}a_{m,\lambda}
                 L_\lambda^{[r]}|\Lambda'\rangle,
\end{equation*}
each coefficient $a_{m,\lambda}$ belongs to $\mathcal R$.
\end{theorem}

In particular, for $r=1$ we write $\beta=\beta_1$ and
$\Lambda=(\Lambda_1,\Lambda_2)$, with $\Lambda_2\ne0$. Then
\begin{equation}
 \Lambda'=(\Lambda_1-\beta,\Lambda_2),\qquad
 \kappa=-2\Delta-\frac{\beta(\Lambda_1-\beta)}{2\Lambda_2}
 \label{eq:rank-one-parameters}
\end{equation}
\cite[Appendix, Section~2a]{Nagoya2015}.

To define the corresponding matrix coefficient, put
$\Vir_{\le0}=\mathbb C C\oplus\bigoplus_{n\le0}\mathbb C L_n$.
For $\Delta_\infty\in\mathbb C$, let
\begin{equation}
 M_{\Delta_\infty}^{*}
 =\mathbb C\langle\Delta_\infty|
     \otimes_{U(\Vir_{\le0})}U(\Vir)
 \label{eq:dual-verma-module}
\end{equation}
be the right Verma module with highest-weight vector
$\langle\Delta_\infty|$ satisfying
\[
 \langle\Delta_\infty|C=c\langle\Delta_\infty|,\qquad
 \langle\Delta_\infty|L_0=\Delta_\infty\langle\Delta_\infty|,
 \qquad \langle\Delta_\infty|L_n=0\quad(n<0).
\]
The superscript $*$ here denotes this right module, rather than the full
algebraic dual.

\begin{definition}
The normalized invariant pairing between $M_{\Delta_\infty}^{*}$ and
$M_{\Lambda'}^{[1]}$ is the bilinear map
\begin{equation}
 \langle\cdot,\cdot\rangle:
 M_{\Delta_\infty}^{*}\times M_{\Lambda'}^{[1]}
 \longrightarrow\mathbb C
 \label{eq:rank-one-pairing}
\end{equation}
determined by
\begin{equation*}
 \langle\langle\Delta_\infty|,|\Lambda'\rangle\rangle=1,
 \qquad
 \langle uL_n,v\rangle=\langle u,L_nv\rangle
 \quad(n\in\mathbb Z).
\end{equation*}
We abbreviate $\langle u,v\rangle$ as $uv$ in bra--ket notation.
\end{definition}

For completeness, the pairing can be constructed directly. Define a
linear functional $\ell$ on the PBW basis of $M_{\Lambda'}^{[1]}$ by
\[
 \ell\bigl(L_{i_1}\cdots L_{i_k}|\Lambda'\rangle\bigr)
 =\begin{cases}
   \Delta_\infty^k,&i_1=\cdots=i_k=0,\\
   0,&\text{some }i_j<0,
  \end{cases}
 \qquad \ell(|\Lambda'\rangle)=1.
\]
It satisfies $\ell(L_0v)=\Delta_\infty\ell(v)$ and
$\ell(L_nv)=0$ for $n<0$.

Thus $\ell$ defines the pairing
\eqref{eq:rank-one-pairing} by
\[
 \bigl\langle\langle\Delta_\infty|X,v\bigr\rangle=\ell(Xv),
 \qquad X\in U(\Vir),\quad v\in M_{\Lambda'}^{[1]}.
\]
This proves existence; invariance and
normalization prove uniqueness. We extend the pairing to formal series
coefficientwise.

\begin{definition}
\label{def:rank-one-block}
The rank-one irregular conformal block considered in this paper is
\begin{equation}
 \mathcal B(z;c,\Delta_\infty,\Delta,\Lambda_1,\Lambda_2,\beta)
 =\langle\Delta_\infty|
    \Phi_{\Lambda',\Lambda}^{\Delta}(z)|\Lambda\rangle,
 \label{eq:rank-one-irregular-block}
\end{equation}
where $\Lambda'$ and $\kappa$ are given by
\eqref{eq:rank-one-parameters}. Equivalently,
\begin{equation}
 \mathcal B(z)=z^\kappa e^{\beta/z}
    \sum_{m=0}^{\infty}b_mz^m,
 \qquad b_m=\langle\Delta_\infty|v_m\rangle,
 \quad b_0=1.
 \label{eq:rank-one-block-expansion}
\end{equation}
\end{definition}

This is the three-point block of type $(0,0,1)$: the two regular
insertions have weights $\Delta_\infty$ and $\Delta$, and the irregular
insertion has rank one. The expansion
\eqref{eq:rank-one-block-expansion} is understood formally; its definition
does not require a convergence assertion.

\subsection{A conjectural combinatorial formula}
\label{subsec:combinatorial-conjecture}

We first record the initial coefficients of the formal series in
\eqref{eq:rank-one-block-expansion}. Thus
\begin{equation*}
 \mathcal B(z)=z^\kappa e^{\beta/z}
     \bigl(1+b_1z+b_2z^2+O(z^3)\bigr).
\end{equation*}
The commutation relations \eqref{eq:irregular-vertex-commutator}, followed
by pairing with $\langle\Delta_\infty|$, give
\begin{align}
 b_1={}&\frac{\beta^3}{4\Lambda_2^2}
       -\frac{3\Lambda_1}{8\Lambda_2^2}\beta^2
       +\left(\frac{\Lambda_1^2}{8\Lambda_2^2}
               -\frac{\Delta+\Delta_\infty}{2\Lambda_2}\right)\beta
       +\frac{\Delta\Lambda_1}{2\Lambda_2},
 \label{eq:block-coefficient-one}\\
 b_2={}&\frac{\beta^6}{32\Lambda_2^4}
       -\frac{3\Lambda_1}{32\Lambda_2^4}\beta^5 \notag\\
 &+\left(\frac{13\Lambda_1^2}{128\Lambda_2^4}
          +\frac{5-4\Delta-4\Delta_\infty}{32\Lambda_2^3}\right)\beta^4
       \notag\\
 &+\left(-\frac{3\Lambda_1^3}{64\Lambda_2^4}
          +\frac{(10\Delta+6\Delta_\infty-10)\Lambda_1}
                 {32\Lambda_2^3}\right)\beta^3 \notag\\
 &+\left(\frac{\Lambda_1^4}{128\Lambda_2^4}
          +\frac{(3-4\Delta-\Delta_\infty)\Lambda_1^2}
                 {16\Lambda_2^3}\right. \notag\\
 &\hspace{1.5em}\left.
          +\frac{\Delta_\infty^2+(2\Delta-3)\Delta_\infty}
                 {8\Lambda_2^2}
          +\frac{6\Delta^2-18\Delta+1-c}{48\Lambda_2^2}
          \right)\beta^2 \notag\\
 &+\left(\frac{(2\Delta-1)\Lambda_1^3}{32\Lambda_2^3}
          +\frac{\bigl(c-1+30\Delta-12\Delta^2
                    +6(1-2\Delta)\Delta_\infty\bigr)\Lambda_1}
                 {48\Lambda_2^2}\right)\beta \notag\\
 &+\frac{\Delta(\Delta-2)\Lambda_1^2}{8\Lambda_2^2}
       +\frac{\Delta\Delta_\infty}{2\Lambda_2}.
 \label{eq:block-coefficient-two}
\end{align}
In particular, $b_1$ is independent of $c$, whereas $b_2$ depends on $c$.

For $c=1$, introduce the parameters
\begin{equation}
 \Delta_\infty=\theta_0^2,\qquad \Delta=\theta_t^2,\qquad
 \Lambda_1=\theta_*,\qquad \Lambda_2=\frac14.
 \label{eq:c-one-parameters}
\end{equation}
Under this specialization, \eqref{eq:block-coefficient-one} and
\eqref{eq:block-coefficient-two} agree with the coefficients of $t^{-1}$
and $t^{-2}$, respectively, in \cite[Equation~(33)]{Nagoya2016}, with
$z=t^{-1}$.
The coefficient $b_1$ can be written as a sum of products of explicit
linear factors:
\begin{equation*}
 b_1=2(\beta-\theta_*)(\beta-\theta_t)(\beta+\theta_t)
       +2\beta(\theta_*-\beta-\theta_0)(\theta_*-\beta+\theta_0).
\end{equation*}
The coefficient $b_2$ admits a similar decomposition, whose explicit
expression we omit. These factorizations motivate the following
conjecture for all coefficients of the irregular conformal block.
To state it, let $\mathbb Y$ be the set of partitions, identified with their Young
diagrams
\[
 \lambda=\{(i,j)\in\mathbb Z_{>0}^2\mid 1\le j\le\lambda_i\}.
\]
Write $\lambda'$ for the conjugate partition and
$h_\lambda(i,j)=\lambda_i-j+\lambda'_j-i+1$ for the hook length.
For $\nu\subseteq\lambda$ and $\eta\subseteq\mu$, put
\begin{align*}
 U_{\lambda/\nu}
  & =\prod_{(i,j)\in\lambda\setminus\nu}
              \bigl(2(\beta-\theta_*)+i-j\bigr),\qquad
 V_{\mu/\eta}
    =\prod_{(i,j)\in\mu\setminus\eta}
              \bigl(-2\beta+i-j\bigr),
 \\
 S_{\lambda,\mu}
  & =(-1)^{|\mu|}
    \prod_{(i,j)\in\lambda}
       \frac{(\beta+i-j)^2-\theta_t^2}{h_\lambda(i,j)^2}
    \prod_{(i,j)\in\mu}
       \frac{(\theta_*-\beta+i-j)^2-\theta_0^2}{h_\mu(i,j)^2}.
\end{align*}
Empty products are understood to be one. The skew-diagram products above
define polynomials even when individual factors vanish.

\begin{conjecture}[Nagoya, {\cite[Conjecture~5.3]{Nagoya2016}}]
\label{conj:c-one-combinatorial-formula}
There exist nonnegative integers $c_{\lambda,\mu}^{\nu,\eta}$,
independent of $\theta_0,\theta_t,\theta_*,\beta$, such that
\begin{equation*}
 \begin{aligned}
 &\mathcal B(z;1,\theta_0^2,\theta_t^2,\theta_*,\tfrac14,\beta)\\
 &\quad=z^{-2\theta_t^2-2\beta(\theta_*-\beta)}e^{\beta/z}
  \sum_{\lambda,\mu\in\mathbb Y}z^{|\lambda|+|\mu|}
  S_{\lambda,\mu}
  \sum_{\substack{\nu\subseteq\lambda,\ \eta\subseteq\mu\\
                   |\nu|=|\eta|}}
       (-1)^{|\nu|}c_{\lambda,\mu}^{\nu,\eta}
       U_{\lambda/\nu}V_{\mu/\eta}.
 \end{aligned}
\end{equation*}
The equality is an identity of formal series after removing the displayed
power and exponential factors.
\end{conjecture}

The proposed coefficients include
\begin{equation*}
 c_{\lambda,\mu}^{\varnothing,\varnothing}=1,
 \qquad c_{\lambda,\mu}^{(1),(1)}=2|\lambda||\mu|
\end{equation*}
whenever the indicated subdiagrams exist
\cite[Section~5.2]{Nagoya2016}. The formula is recalled here in its
original conjectural form.

\section{Screening integral}
\label{sec:screening-integral}

\subsection{Free field realization}
\label{subsec:free-field-realization}

The Heisenberg algebra is the complex Lie algebra
\[
 \mathfrak h=\bigoplus_{n\in\mathbb Z}\mathbb C a_n\oplus\mathbb C K
\]
with relations
\begin{equation*}
 [a_m,a_n]=m\delta_{m+n,0}K,\qquad [K,a_n]=0.
\end{equation*}
We work at level one, so that $K$ acts as the identity. For $p\in\mathbb C$,
the Fock module of momentum $p$ is
\begin{equation*}
 \mathcal F_p
 =U(\mathfrak h)\otimes_{U(\mathfrak h_{\ge0})}\mathbb C|p\rangle
 =\mathbb C[a_{-1},a_{-2},\ldots]|p\rangle,
\end{equation*}
where $\mathfrak h_{\ge0}=\mathbb C K\oplus
\bigoplus_{n\ge0}\mathbb C a_n$ and
\[
 K|p\rangle=|p\rangle,\qquad a_0|p\rangle=p|p\rangle,
 \qquad a_n|p\rangle=0\quad(n>0).
\]
For a partition $\mu$, write $a_{-\mu}=a_{-\mu_1}\cdots a_{-\mu_k}$.
Assign degree $|\mu|$ to $a_{-\mu}|p\rangle$. This gives the grading and
its associated filtration
\begin{equation}
 \mathcal F_p=\bigoplus_{m\ge0}\mathcal F_p[m],\qquad
 \mathcal F_{p,\le m}=\bigoplus_{j=0}^{m}\mathcal F_p[j]
 =\operatorname{span}_{\mathbb C}
    \{a_{-\mu}|p\rangle\mid |\mu|\le m\}.
 \label{eq:fock-filtration}
\end{equation}
We use the completion $\widehat{\mathcal F}_p=
\prod_{m\ge0}\mathcal F_p[m]$ to accommodate coherent states.
For $\ell\in\mathbb Z_{\ge0}$, let
\begin{equation}
 \pi_{\le\ell}:\widehat{\mathcal F}_p
       \longrightarrow\mathcal F_{p,\le\ell},
 \qquad (v_d)_{d\ge0}\longmapsto\sum_{d=0}^{\ell}v_d,
 \qquad v_d\in\mathcal F_p[d],
 \label{eq:fock-level-projection}
\end{equation}
be the projection associated with the grading
\eqref{eq:fock-filtration}.

Let $r\ge1$ and $\boldsymbol c=(c_1,\ldots,c_r)$, with $c_r\ne0$.
Define the coherent state
\begin{equation}
 |p;\boldsymbol c\rangle
 =\exp\left(\sum_{j=1}^{r}\frac{c_j}{j}a_{-j}\right)|p\rangle
 \in\widehat{\mathcal F}_p.
 \label{eq:coherent-state}
\end{equation}
The Heisenberg relations imply
\begin{equation*}
 a_0|p;\boldsymbol c\rangle=p|p;\boldsymbol c\rangle,\qquad
 a_n|p;\boldsymbol c\rangle=
 \begin{cases}
  c_n|p;\boldsymbol c\rangle,&1\le n\le r,\\
  0,&n>r.
 \end{cases}
\end{equation*}
The Heisenberg module generated by this state and its filtration are
\begin{align*}
 \mathcal F_{p,\boldsymbol c}^{[r]}
  &=\mathbb C[a_{-1},a_{-2},\ldots]|p;\boldsymbol c\rangle,
 \\
 \mathcal F_{p,\boldsymbol c,\le m}^{[r]}
  &=\operatorname{span}_{\mathbb C}
    \{a_{-\mu}|p;\boldsymbol c\rangle\mid |\mu|\le m\}.
\end{align*}
Here the filtration measures the degree of the polynomial multiplying
the coherent state, rather than the degree of that state in
$\widehat{\mathcal F}_p$.

Choose $\rho\in\mathbb C$ such that $c=1-12\rho^2$. The free field
realization of the Virasoro algebra is 
\begin{equation}
 L_n=\frac12\sum_{k\in\mathbb Z}:a_{n-k}a_k:
         -\rho(n+1)a_n.
 \label{eq:free-field-virasoro}
\end{equation}
Normal ordering places the modes $a_{-j}$, $j>0$, to the left of the
positive modes. The operators $L_n$ satisfy
\eqref{eq:virasoro-relations} with central charge $c$, and
$L_0|p\rangle=\frac12p(p-2\rho)|p\rangle$.
Set $c_0=p$ and $c_j=0$ for $j<0$ or $j>r$. Then
\begin{equation}
 L_n|p;\boldsymbol c\rangle=
 \begin{cases}
  \Lambda_n|p;\boldsymbol c\rangle,&r\le n\le2r,\\
  0,&n>2r,
 \end{cases}
 \qquad
 \Lambda_n=\frac12\sum_{j=0}^{n}c_jc_{n-j}
                -\rho(n+1)c_n.
 \label{eq:coherent-virasoro-eigenvalues}
\end{equation}
In particular, $\Lambda_{2r}=c_r^2/2\ne0$. The induced Virasoro
homomorphism
\begin{equation}
 \iota_{p,\boldsymbol c}:M_\Lambda^{[r]}
     \longrightarrow\mathcal F_{p,\boldsymbol c}^{[r]},\qquad
 |\Lambda\rangle\longmapsto|p;\boldsymbol c\rangle
 \label{eq:irregular-fock-identification}
\end{equation}
is an isomorphism preserving the filtrations. To see this, let $j\ge1$
and let $\nu$ be a partition. The free field formula
\eqref{eq:free-field-virasoro} gives
\[
 L_{r-j}a_{-\nu}|p;\boldsymbol c\rangle
 =c_r a_{-j}a_{-\nu}|p;\boldsymbol c\rangle+w,
 \qquad w\in\mathcal F_{p,\boldsymbol c,\le |\nu|+j-1}^{[r]}.
\]
Thus the term of highest degree is obtained by multiplying
$a_{-\nu}$ by $c_r a_{-j}$; all remaining terms have strictly smaller
degree. Applying this identity successively yields
\begin{equation*}
 \iota_{p,\boldsymbol c}
    \bigl(L_\mu^{[r]}|\Lambda\rangle\bigr)
 =c_r^{\ell(\mu)}a_{-\mu}|p;\boldsymbol c\rangle
       +\text{terms of degree less than }|\mu|,
\end{equation*}
where $\ell(\mu)$ is the length of $\mu$. Since $c_r\ne0$, this
triangular change of basis is invertible at every finite degree. Thus
$U_m$ corresponds to $\mathcal F_{p,\boldsymbol c,\le m}^{[r]}$ under
the free field realization.

We next introduce a formal zero mode $q$, satisfying
$[a_0,q]=1$ and $[a_n,q]=0$ for $n\ne0$. Its exponential denotes the
momentum-shift map
\begin{equation*}
 e^{\gamma q}:\mathcal F_p\longrightarrow\mathcal F_{p+\gamma},
 \qquad e^{\gamma q}a_{-\mu}|p\rangle
       =a_{-\mu}|p+\gamma\rangle.
\end{equation*}
It extends coefficientwise to the completions and sends
$|p;\boldsymbol c\rangle$ to $|p+\gamma;\boldsymbol c\rangle$.
The free boson field is
\begin{equation*}
 \phi(z)=q+a_0\log z-\sum_{n\ne0}\frac{a_n}{n}z^{-n}.
\end{equation*}
Its normally ordered exponential is defined by
\begin{equation*}
 V_\gamma(z)=:e^{\gamma\phi(z)}:
 =\exp\left(\gamma\sum_{j\ge1}\frac{a_{-j}}{j}z^j\right)
   e^{\gamma q}z^{\gamma a_0}
   \exp\left(-\gamma\sum_{j\ge1}\frac{a_j}{j}z^{-j}\right).
\end{equation*}
The order of the zero-mode factors in this formula is part of our
convention. A direct calculation gives
\begin{equation}
 [L_n,V_\gamma(z)]
 =z^n\left(z\frac{d}{dz}+(n+1)\Delta(\gamma)\right)V_\gamma(z),
 \qquad \Delta(\gamma)=\frac12\gamma(\gamma-2\rho).
 \label{eq:free-field-vertex-commutator}
\end{equation}

The action of $V_\gamma(z)$ on the coherent state is
\begin{equation*}
 V_\gamma(z)|p;\boldsymbol c\rangle
 =z^{\gamma p}
   \exp\left(-\gamma\sum_{j=1}^{r}\frac{c_j}{j}z^{-j}\right)
   \exp\left(\gamma\sum_{j\ge1}\frac{a_{-j}}{j}z^j\right)
        |p+\gamma;\boldsymbol c\rangle.
\end{equation*}
The vector-valued series
\[
 \exp\left(\gamma\sum_{j\ge1}\frac{a_{-j}}{j}z^j\right)
 |p+\gamma;\boldsymbol c\rangle
\]
has constant term $|p+\gamma;\boldsymbol c\rangle$, and its coefficient
of $z^m$ belongs to $\mathcal F_{p+\gamma,\boldsymbol c,\le m}^{[r]}$.
The target Virasoro parameters obey
\begin{equation}
 \Lambda'_n=\Lambda_n+\delta_{n,r}\gamma c_r
           =\Lambda_n-\delta_{n,r}r\beta_r,
 \qquad \beta_j=-\frac{\gamma c_j}{j},\quad
 \kappa=\gamma p.
 \label{eq:free-field-irregular-parameters}
\end{equation}
Hence, under the identifications \eqref{eq:irregular-fock-identification},
$V_\gamma(z)$ is the normalized rank-zero irregular vertex operator
$\Phi_{\Lambda',\Lambda}^{\Delta(\gamma)}(z)$ of
Section~\ref{subsec:irregular-vertex-operators}. In this construction,
the conformal weight and the irregular parameters are determined by
the free field parameters through
\eqref{eq:coherent-virasoro-eigenvalues},
\eqref{eq:free-field-vertex-commutator}, and
\eqref{eq:free-field-irregular-parameters}.

For rank one, the parameter identification reads
\begin{equation}
 \Lambda_1=(p-2\rho)c_1,\qquad
 \Lambda_2=\frac12c_1^2,\qquad
 \beta=-\gamma c_1,\qquad \kappa=\gamma p,
 \label{eq:rank-one-free-field-parameters}
\end{equation}
which is consistent with \eqref{eq:rank-one-parameters}.
Since $\Delta=\frac12\gamma(\gamma-2\rho)$ and
$\beta=-\gamma c_1$, fixing $\rho$ and $c_1$ prevents $\Delta$ and
$\beta$ from being chosen independently. This restriction motivates
the introduction of screening integrals.

\subsection{The screening current and the transformed integral}
\label{subsec:screening-integral}

Choose $\gamma_+\in\mathbb C$ such that
$\frac12\gamma_+(\gamma_+-2\rho)=1$. The screening current is the
vertex operator of conformal weight one
\begin{equation*}
 S_+(t)=V_{\gamma_+}(t)=:e^{\gamma_+\phi(t)}:.
\end{equation*}
By \eqref{eq:free-field-vertex-commutator}, it satisfies
\begin{equation}
 [L_k,S_+(t)]
 =\frac{d}{dt}\bigl(t^{k+1}S_+(t)\bigr),\qquad k\in\mathbb Z.
 \label{eq:screening-current-commutator}
\end{equation}

For the screening integral, denote the vertex charge by $\gamma_z$, so that
the insertion is $V_{\gamma_z}(z)$. We now restrict to rank one and write $|p;c_1\rangle$ for the
coherent state \eqref{eq:coherent-state}. Let $m,n\in\mathbb Z_{\ge0}$,
and define the total momentum $p_{\mathrm{out}}$ and the parameter
$\gamma_\infty$ by
\begin{equation}
 p_{\mathrm{out}}=p+\gamma_z+(m+n)\gamma_+,
 \qquad \gamma_\infty=2\rho-p_{\mathrm{out}}.
 \label{eq:screening-neutrality}
\end{equation}
Related confluent screening-integral constructions in conformal field theory
were studied by the author and Sun \cite{NagoyaSun2010}.
More recently, Adachi \cite{Adachi2026} analyzed double integrals of confluent
Selberg type, obtaining sectorial asymptotic expansions and explicit Stokes
matrices for the confluent Dotsenko--Fateev system.

For $z=re^{\mathrm i\varphi}$, where $r>0$, choose
$\psi\in\mathbb R$ such that $0<|\varphi-\psi|<\pi$.
We consider the vector-valued screening integral
\begin{equation*}
 \begin{split}
 \mathcal S_{m,n}(z)={}&
 \int_{\Gamma_\xi^{(m)}(z)}\int_{\Gamma_\eta^{(n)}(\psi)}
 S_+(\eta_n)\cdots S_+(\eta_1)V_{\gamma_z}(z)\\
 &\hspace{18mm}\times S_+(\xi_m)\cdots S_+(\xi_1)
 |p;c_1\rangle\,d^m\xi\,d^n\eta.
 \end{split}
\end{equation*}
Normal ordering gives the explicit expression
\begin{equation}
 \begin{aligned}
 \mathcal S_{m,n}(z)
 ={}&z^{p\gamma_z}e^{-\gamma_z c_1/z}
 \int_{\Gamma_\xi^{(m)}(z)}\int_{\Gamma_\eta^{(n)}(\psi)}
 \prod_{i<i'}(\xi_{i'}-\xi_i)^{\gamma_+^2}
 \prod_{j<j'}(\eta_{j'}-\eta_j)^{\gamma_+^2}\\
 &\quad\times\prod_{i=1}^m\prod_{j=1}^n
       (\eta_j-\xi_i)^{\gamma_+^2}
 \prod_{i=1}^m\xi_i^{\gamma_+p}
       (z-\xi_i)^{\gamma_+\gamma_z}e^{-\gamma_+c_1/\xi_i}\\
 &\quad\times\prod_{j=1}^n\eta_j^{\gamma_+p}
       (\eta_j-z)^{\gamma_+\gamma_z}e^{-\gamma_+c_1/\eta_j}\\
 &\quad\times\exp\left[
 \sum_{s\ge1}\frac{a_{-s}}s
 \left(\gamma_z z^s+\gamma_+\sum_{i=1}^m\xi_i^s
                   +\gamma_+\sum_{j=1}^n\eta_j^s\right)\right]
 |p_{\mathrm{out}};c_1\rangle\,d^m\xi\,d^n\eta.
 \end{aligned}
 \label{eq:screening-integral-normal-ordered}
\end{equation}
Here the ordered integration contours are
\begin{equation}
 \begin{aligned}
 \Gamma_\xi^{(m)}(z)
 &=\left\{(t_1z,\ldots,t_mz)\in\mathbb C^m
       \mid 0<t_1<\cdots<t_m<1\right\},\\
 \Gamma_\eta^{(n)}(\psi)
 &=\left\{(e^{\mathrm i\psi}s_1,\ldots,e^{\mathrm i\psi}s_n)
       \in\mathbb C^n\mid 0<s_1<\cdots<s_n<\infty\right\}.
 \end{aligned}
 \label{eq:screening-original-contours}
\end{equation}

Make the changes of variables
\begin{equation}
 \xi_i=\frac{z}{1-zx_i},\qquad \eta_j=\frac1{y_j}.
 \label{eq:screening-variable-change}
\end{equation}
After reversing the order within each set of variables, the
transformed contours are the fixed chambers
\begin{equation}
 \begin{aligned}
 \Gamma_x^{(m)}(\varphi)
 &=\left\{(-e^{-\mathrm i\varphi}u_1,\ldots,
           -e^{-\mathrm i\varphi}u_m)\in\mathbb C^m
       \mid 0<u_1<\cdots<u_m<\infty\right\},\\
 \Gamma_y^{(n)}(\psi)
 &=\left\{(e^{-\mathrm i\psi}v_1,\ldots,
           e^{-\mathrm i\psi}v_n)\in\mathbb C^n
       \mid 0<v_1<\cdots<v_n<\infty\right\}.
 \end{aligned}
 \label{eq:screening-transformed-contours}
\end{equation}
Absorbing the constant factors into a nonzero constant
$\varepsilon_{m,n}$, we obtain
\begin{equation}
 \mathcal S_{m,n}(z)
 =\varepsilon_{m,n}z^\kappa e^{\beta/z}
 \int_{\Gamma_x^{(m)}(\varphi)}\int_{\Gamma_y^{(n)}(\psi)}
 W(x,y)G(z;x,y)\mathcal E(z;x,y)\,d^mx\,d^ny,
 \label{eq:screening-integral-transformed}
\end{equation}
where
\begin{equation*}
 \begin{aligned}
 \beta&=-c_1(\gamma_z+m\gamma_+),\\
 \kappa&=(\gamma_z+m\gamma_+)p
       +2m(\gamma_+\gamma_z+1)+m(m-1)\gamma_+^2,
 \end{aligned}
\end{equation*}
and the three factors in the integrand are
\begin{equation}
 W(x,y)=
 \prod_{i<i'}(x_{i'}-x_i)^{\gamma_+^2}
 \prod_{j<j'}(y_{j'}-y_j)^{\gamma_+^2}
 \prod_{i=1}^m x_i^{\gamma_+\gamma_z}e^{\gamma_+c_1x_i}
       \prod_{j=1}^n y_j^{\gamma_+\gamma_\infty}
                         e^{-\gamma_+c_1y_j},
 \label{eq:screening-weight}
\end{equation}
\begin{equation}
 G(z;x,y)=
 \prod_{i=1}^m(1-zx_i)^{\gamma_+\gamma_\infty}
 \prod_{j=1}^n(1-zy_j)^{\gamma_+\gamma_z}
 \prod_{i=1}^m\prod_{j=1}^n
          \bigl(1-z(x_i+y_j)\bigr)^{\gamma_+^2},
 \label{eq:screening-scalar-factor}
\end{equation}
\begin{equation}
 \mathcal E(z;x,y)=
 \exp\left[\sum_{s\ge1}\frac{a_{-s}}s
 \left\{\gamma_+\sum_{j=1}^n y_j^{-s}
 +z^s\left(\gamma_z+\gamma_+\sum_{i=1}^m(1-zx_i)^{-s}\right)
 \right\}\right]
 |p_{\mathrm{out}};c_1\rangle.
 \label{eq:screening-creation-factor}
\end{equation}
In particular, $W$ is independent of $z$, $G$ is scalar-valued,
and $\mathcal E$ contains the creation operators and the target
coherent state. To check the exponents, note that
\begin{equation*}
 -\gamma_+(p+\gamma_z)-2-(m+n-1)\gamma_+^2
 =\gamma_+\gamma_\infty,
\end{equation*}
which follows from \eqref{eq:screening-neutrality} and
$\gamma_+^2-2\rho\gamma_+=2$. This is both the exponent of
$1-zx_i$ and the monomial exponent of $y_j$ after the Jacobians
are included. The factor $e^{\beta/z}$ follows from
$e^{-\gamma_+c_1/\xi_i}
=e^{-\gamma_+c_1/z}e^{\gamma_+c_1x_i}$.

For the convergence and asymptotic estimates below, fix
$N,\ell\in\mathbb Z_{\ge0}$ and project the screening integrand to
$\mathcal F_{p_{\mathrm{out}},\le\ell}$.
Here $N$ is the desired order of expansion in $z$.
We use the following sufficient conditions
\begin{equation}
 \begin{aligned}
 &\operatorname{Re}(\gamma_+^2)>0,
 \qquad \operatorname{Re}(\gamma_+\gamma_z)>0,\\
 &\operatorname{Re}(\gamma_+c_1e^{-\mathrm i\varphi})>0,
 \qquad \operatorname{Re}(\gamma_+c_1e^{-\mathrm i\psi})>0,\\
 &\operatorname{Re}(\gamma_+\gamma_\infty)>\ell+N+1.
 \end{aligned}
 \label{eq:screening-convergence-domain}
\end{equation}
The two conditions involving $c_1$ give exponential decay at
infinity on the transformed contours. The remaining conditions
ensure integrability at the endpoints and collisions, even in the presence
of the inverse powers of $y_j$ arising from the creation operators at Fock
levels up to the cutoff.
Keeping the contours in \eqref{eq:screening-transformed-contours}
fixed, we use \eqref{eq:screening-integral-transformed} to continue
$\mathcal S_{m,n}(z)$ to a sufficiently small sector about
$\arg z=\varphi$.
This corresponds to deforming the original $\xi$-contours as $z$
varies in the sector.

\subsection{Multivariable contour estimates and boundary vanishing}
\label{subsec:contour-estimates}

The ordinary Fock cutoff bounds the total degree in the inverse
variables $y_j^{-1}$ arising from the creation factor, allowing us
to establish convergence and remainder estimates.
For each Fock momentum $p'$ and $\ell\in\mathbb Z_{\ge0}$, we equip
the finite-dimensional space $\mathcal F_{p',\le\ell}$ with a fixed
norm. Remainder estimates at cutoff $\ell$ are understood in this
norm after applying $\pi_{\le\ell}$ from
\eqref{eq:fock-level-projection}.

Set
\begin{equation*}
 R=1+\sum_{i=1}^m u_i+\sum_{j=1}^n v_j,
 \qquad S=R-1.
\end{equation*}
Unless otherwise indicated, all integrals below are over
$\Gamma_x^{(m)}(\varphi)\times\Gamma_y^{(n)}(\psi)$, defined in
\eqref{eq:screening-transformed-contours}.
We show that the estimates below are uniform on compact parameter
sets satisfying the strict inequalities in \eqref{eq:screening-convergence-domain}
and on closed subsectors
of a sufficiently small sector about $\arg z=\varphi$. Empty sets of screening variables
contribute one and require no endpoint or decay condition.

\begin{lemma}[Absolute convergence and finite-order remainder]
\label{lem:contour-remainder}
Assume \eqref{eq:screening-convergence-domain} and fix
$M,h\in\mathbb Z_{\ge0}$ with $h\le\ell+N+2$.
The transformed integral \eqref{eq:screening-integral-transformed}
is absolutely convergent at Fock levels at most $h$. There are
$\mathcal F_{p_{\mathrm{out}},\le h}$-valued Laurent polynomials
$T_{q,h}(x,y)$ such that
\begin{equation}
 \begin{aligned}
 F_h(z;x,y)&=\pi_{\le h}\bigl(G(z;x,y)\mathcal E(z;x,y)\bigr)\\
 &=\sum_{q=0}^M z^qT_{q,h}(x,y)+\mathcal R_{M,h}(z;x,y),
 \end{aligned}
 \label{eq:contour-taylor-remainder}
\end{equation}
and
\begin{equation}
 \int |W|\,\|\mathcal R_{M,h}\|\,|d^mx\,d^ny|
 =O(|z|^{M+1}).
 \label{eq:contour-integrated-remainder}
\end{equation}
More precisely, put
\[
 B_h(y)=\prod_{j=1}^n(1+|y_j|^{-1})^h.
\]
For a sufficiently small fixed $\delta>0$, one has
\[
 \|\mathcal R_{M,h}(z;x,y)\|
 \le C|z|^{M+1}R^{M+1}B_h(y)
 \quad\text{on }R\le\delta/|z|.
\]
The weighted integral of the remainder over $R>\delta/|z|$ is
$O(e^{-\gamma/|z|})$ for some $\gamma>0$.
\end{lemma}

\begin{proof}
If $m=n=0$, then $F_h$ is a polynomial in $z$, and the assertions
follow directly. Otherwise, let $\sigma>0$ be a common lower bound
for the exponential decay rates
\[
 \operatorname{Re}(\gamma_+c_1e^{-\mathrm i\varphi}),
 \qquad
 \operatorname{Re}(\gamma_+c_1e^{-\mathrm i\psi}),
\]
omitting the first rate if $m=0$ and the second if $n=0$.
We first establish the following weighted moment bound for every $D\ge0$:
\begin{equation}
 \int |W(x,y)|B_h(y)R^D\,|d^mx\,d^ny|<\infty.
 \label{eq:contour-moment-majorant}
\end{equation}
On the integration contours,
\[
 B_h(y)=\prod_{j=1}^{n}v_j^{-h}(1+v_j)^h.
\]
Thus, in \eqref{eq:contour-moment-majorant}, the factor $v_j^{-h}$
changes the exponent of $v_j$ in $|W|$ from
$\operatorname{Re}(\gamma_+\gamma_\infty)$ to
$\operatorname{Re}(\gamma_+\gamma_\infty)-h$.
The inequalities
\[
 \operatorname{Re}(\gamma_+\gamma_z)>-1,
 \qquad
 \operatorname{Re}(\gamma_+\gamma_\infty)-h>-1,
 \qquad
 \operatorname{Re}(\gamma_+^2)>0
\]
follow from \eqref{eq:screening-convergence-domain} and
$h\le\ell+N+2$. They ensure local integrability at the origins
and collisions. At infinity, the exponential decay of $|W|$
controls the polynomial growth from $R^D$, the factors $(1+v_j)^h$,
and the power factors in $|W|$. This proves
\eqref{eq:contour-moment-majorant}. The same argument applies with
$|W|$ replaced by $e^{\sigma S/2}|W|$, since exponential decay remains.

For fixed $x,y$, we have $|x_i|,|y_j|,|x_i+y_j|<R$.
Hence, on $|w|<1/(2R)$, all factors $1-wx_i$, $1-wy_j$, and
$1-w(x_i+y_j)$ lie in the disc $|\zeta-1|<1/2$.
The complex powers in \eqref{eq:screening-scalar-factor} are
holomorphic there, using logarithms that vanish at $w=0$.
At cutoff $h$, only finitely many terms of the creation exponential
in \eqref{eq:screening-creation-factor} contribute, so $F_h(w;x,y)$
is holomorphic for $|w|<1/(2R)$.
On $|w|=1/(4R)$, the above factors have absolute value at least $3/4$,
and $|w/(1-wx_i)|\le 1/(3R)$.
Thus $G$ and the $w$-dependent creation coefficients are uniformly
bounded. The inverse $y$-monomials at cutoff $h$ have total degree
at most $h$, and each is bounded in absolute value by $B_h(y)$. Consequently,
\[
 \sup_{|w|=1/(4R)}\|F_h(w;x,y)\|\le CB_h(y)
\]
for some constant $C>0$ independent of $x,y$.
Set $T_{q,h}(x,y)=[w^q]F_h(w;x,y)$. These coefficients are
Laurent polynomials in $x,y$, since only finitely many terms contribute
at fixed $q,h$. Cauchy's estimate gives
\begin{equation}
 \|T_{q,h}(x,y)\|\le C_qR^qB_h(y)
 \label{eq:contour-coefficient-bound}
\end{equation}
with $C_q=C4^q$. Summing the Taylor tail gives, for $4|z|R<1$,
\[
 \begin{aligned}
 \|\mathcal R_{M,h}(z;x,y)\|
 &\le CB_h(y)\sum_{q=M+1}^{\infty}(4|z|R)^q\\
 &=CB_h(y)\frac{(4|z|R)^{M+1}}{1-4|z|R}.
 \end{aligned}
\]
For $0<\delta<1/8$, the denominator is bounded below by $1/2$
on $R\le\delta/|z|$. Thus
\[
 \|\mathcal R_{M,h}(z;x,y)\|
 \le C|z|^{M+1}R^{M+1}B_h(y)
 \qquad\text{on }R\le\delta/|z|.
\]
By \eqref{eq:contour-moment-majorant}, the integrals
$\int W(x,y)T_{q,h}(x,y)\,d^mx\,d^ny$
converge absolutely, and
\[
 \int_{R\le\delta/|z|}
 |W(x,y)|\,\|\mathcal R_{M,h}(z;x,y)\|\,|d^mx\,d^ny|
 =O(|z|^{M+1}).
\]

For the complementary region, put $r=|z|$ and assume $r\le1$.
The triangle inequality and angular separation give
\[
 \begin{aligned}
 c(1+ru_i)&\le |1-zx_i|\le1+ru_i\le R,\\
 c(1+rv_j)&\le |1-zy_j|\le1+rv_j\le R,\\
 c\bigl(1+r(u_i+v_j)\bigr)
 &\le |1-z(x_i+y_j)|\le1+r(u_i+v_j)\le R
 \end{aligned}
\]
throughout a sufficiently small sector about $\arg z=\varphi$,
for some constant $c>0$ independent of $x,y,z$.
These factors are therefore nonzero, and angular separation also
keeps their arguments uniformly bounded.
All exponents in $G$ have positive real parts by
\eqref{eq:screening-convergence-domain}, so $G$ has at most
polynomial growth in $R$. Moreover, the lower bound for $|1-zx_i|$ gives
$|z/(1-zx_i)|\le C|z|$, so the $z$-dependent creation coefficients
at cutoff $h$ are uniformly bounded. Together with
\eqref{eq:contour-coefficient-bound}, this yields
\begin{equation}
 \|F_h(z;x,y)\|
 +\left\|\sum_{q=0}^{M}z^qT_{q,h}(x,y)\right\|
 \le CB_h(y)R^D
 \label{eq:contour-polynomial-majorant}
\end{equation}
for some fixed $D\ge0$. In particular,
\eqref{eq:contour-moment-majorant} also proves the absolute
convergence of the projected integral.

If $R>\delta/|z|$ and $|z|$ is sufficiently small, then
$S>\delta/(2|z|)$. Hence
\[
 \begin{aligned}
 &\int_{R>\delta/|z|}
 |W|\,\|\mathcal R_{M,h}(z;x,y)\|\,|d^mx\,d^ny|\\
 &\qquad\le C e^{-\sigma\delta/(4|z|)}
 \int e^{\sigma S/2}|W|B_h(y)R^D\,|d^mx\,d^ny|\\
 &\qquad=O\!\left(e^{-\sigma\delta/(4|z|)}\right).
 \end{aligned}
\]
Combining this exponential tail estimate with the bound on
$R\le\delta/|z|$ proves \eqref{eq:contour-integrated-remainder}.
For fixed $M,h$, all bounds are uniform on compact parameter sets satisfying
\eqref{eq:screening-convergence-domain} and on closed subsectors
of the chosen sector.
Indeed, the exponential decay rates have a common positive lower bound,
the endpoint exponents stay uniformly above $-1$, and angular separation
is uniform on these sets.
\end{proof}

\begin{lemma}[Differentiation and input cutoff]
\label{lem:contour-differentiation}
Assume \eqref{eq:screening-convergence-domain} and fix
$M,h\in\mathbb Z_{\ge0}$ with $h\le\ell+N+2$.
Then the following assertions hold.
\begin{enumerate}
\item[(i)] The integral
$\int W(x,y)F_h(z;x,y)\,d^mx\,d^ny$
can be differentiated under the integral sign by $z\partial_z$.

\item[(ii)] The remainder $\mathcal R_{M,h}$ in
\eqref{eq:contour-taylor-remainder} satisfies
\begin{equation*}
 \int |W|\,\|z\partial_z\mathcal R_{M,h}\|\,|d^mx\,d^ny|
 =O(|z|^{M+1}).
\end{equation*}

\item[(iii)] For $1\le k\le N+2$, the remainder
$\mathcal R_{M,\ell+k}$ at cutoff $\ell+k$ satisfies
\[
 \int |W|\,\bigl\|\pi_{\le\ell}L_k\mathcal R_{M,\ell+k}\bigr\|
 \,|d^mx\,d^ny|=O(|z|^{M+1}).
\]
\end{enumerate}
\end{lemma}

\begin{proof}
On the fixed transformed contours, $F_h(z;x,y)$ is holomorphic
in $z$ on the sector of Lemma~\ref{lem:contour-remainder}.
The bounds \eqref{eq:contour-moment-majorant} and
\eqref{eq:contour-polynomial-majorant} provide an integrable majorant
for $WF_h$, locally uniform in $z$. Hence the projected integral
is holomorphic, and
\[
 z\partial_z\int W(x,y)F_h(z;x,y)\,d^mx\,d^ny
 =\int W(x,y)\,z\partial_zF_h(z;x,y)\,d^mx\,d^ny.
\]

Fix a closed subsector and choose $0<\eta<1$ so that, for all
sufficiently small $z$ in this subsector, the disks
$|w-z|\le\eta|z|$ lie in a larger closed subsector.
Cauchy's formula and Tonelli's theorem give
\[
 \begin{aligned}
 &\int |W|\,\|z\partial_z\mathcal R_{M,h}(z;x,y)\|
       \,|d^mx\,d^ny|\\
 &\qquad\le \frac1\eta
       \sup_{|w-z|=\eta|z|}
       \int |W|\,\|\mathcal R_{M,h}(w;x,y)\|
       \,|d^mx\,d^ny|
       =O(|z|^{M+1}),
 \end{aligned}
\]
where the last estimate follows from
\eqref{eq:contour-integrated-remainder}, uniformly on the larger
subsector, and $|w|\le(1+\eta)|z|$.

Finally, the restriction of $\pi_{\le\ell}L_k$ to
$\mathcal F_{p_{\mathrm{out}},\le\ell+k}$ is a bounded linear map
between finite-dimensional spaces. Since $\ell+k\le\ell+N+2$,
\eqref{eq:contour-integrated-remainder} applies at cutoff $\ell+k$.
\end{proof}

For later use, we record the cutoff relation
\begin{equation}
 \pi_{\le\ell}L_k
 =\pi_{\le\ell}L_k\pi_{\le\ell+k},
 \qquad k\ge1.
 \label{eq:virasoro-finite-input-cutoff}
\end{equation}
Indeed, \eqref{eq:free-field-virasoro} shows that $L_k$ lowers
the Fock level by $k$, so input levels greater than $\ell+k$
do not contribute after projection to levels at most $\ell$.

\begin{lemma}[Boundary vanishing and finite Ward identities]
\label{lem:contour-boundary-vanishing}
Fix $\ell\ge0$ and a finite mode bound $K\ge1$. Assume
\eqref{eq:screening-convergence-domain} with $N+2\ge K$.
Then the following assertions hold.
\begin{enumerate}
\item[(i)] At output cutoff $\ell$, the sum of the total-derivative terms
arising from \eqref{eq:screening-current-commutator}
has zero integral for $1\le k\le K$.

\item[(ii)] The screening integral satisfies
\begin{equation}
 \begin{aligned}
 &\pi_{\le\ell}
 (L_k-\delta_{k,1}\Lambda_1-\delta_{k,2}\Lambda_2)
       \mathcal S_{m,n}(z)\\
 &\qquad=\pi_{\le\ell}z^k
       \bigl(z\partial_z+(k+1)\Delta(\gamma_z)\bigr)
       \mathcal S_{m,n}(z),\qquad 1\le k\le K,
 \end{aligned}
 \label{eq:screening-finite-ward-identity}
\end{equation}
where $\Lambda_1=(p-2\rho)c_1$ and $\Lambda_2=c_1^2/2$ are the
source eigenvalues in \eqref{eq:rank-one-free-field-parameters}.

\item[(iii)] Put $p_0=p+\gamma_z+m\gamma_+$. For every
$J\in\mathbb Z_{\ge0}$, the fixed semi-infinite screening operator
\begin{equation*}
 Q_{+,\psi}^{(n)}u
 =\int_{\Gamma_\eta^{(n)}(\psi)}
 S_+(\eta_n)\cdots S_+(\eta_1)u\,d^n\eta
\end{equation*}
intertwines the Virasoro operators $L_k$, $-J\le k\le K$,
at output cutoff $\ell$:
\begin{equation}
 \pi_{\le\ell}L_kQ_{+,\psi}^{(n)}u
 =\pi_{\le\ell}Q_{+,\psi}^{(n)}L_ku,
 \qquad -J\le k\le K,
 \label{eq:fixed-screening-intertwining}
\end{equation}
for every finite coherent descendant
$u\in\mathcal F_{p_0,c_1}^{[1]}$.
\end{enumerate}
\end{lemma}

\begin{proof}
\textup{(i)} We work coefficientwise at output cutoff $\ell$, with $z$ in a
compact subset of the punctured sector. Let $\mathcal D(\xi,\eta;z)$
denote a coefficient of the normal-ordered density in
\eqref{eq:screening-integral-normal-ordered} and write
\[
 \widetilde{\mathcal D}(x,y;z)
 =\mathcal D(\xi(x,z),\eta(y);z)\mathcal J(x,y;z),
\]
where $\mathcal J$ is the full Jacobian of
\eqref{eq:screening-variable-change}. We integrate first over
compact truncations of the ordered chambers.
Lemmas~\ref{lem:contour-remainder}--\ref{lem:contour-differentiation}
supply the integral bounds, and \eqref{eq:virasoro-finite-input-cutoff}
gives the input cutoff $\ell+K$.

For a screening variable on the fixed ray, the substitution
$\eta_j=1/y_j$ gives
\[
 \mathcal J\,\partial_{\eta_j}
 \bigl(\eta_j^{k+1}\mathcal D\bigr)
 =-\partial_{y_j}\bigl(y_j^{1-k}\widetilde{\mathcal D}\bigr).
\]
We therefore estimate $y_j^{1-k}\widetilde{\mathcal D}$ at the boundary.
The creation factor
\eqref{eq:screening-creation-factor} contributes inverse powers of
total degree at most $\ell$. By \eqref{eq:screening-convergence-domain}
and $N+2\ge K$, the real part of the exponent of $v_j$ in this
expression has the positive lower bound
\[
 \operatorname{Re}(\gamma_+\gamma_\infty)-\ell+1-k
 \ge\operatorname{Re}(\gamma_+\gamma_\infty)-\ell+1-K>0.
\]
For the $\xi_i$-terms we estimate $\xi_i^{k+1}\mathcal D$.
At $\xi_i=z$ and at collisions within either family away from the
endpoints, the relevant factors have exponents $\gamma_+\gamma_z$
and $\gamma_+^2$, respectively, whose real parts are also positive.
Together with the bound above at $v_j=0$, this implies vanishing
of the boundary terms. Differentiating an algebraic factor lowers
its exponent by one, leaving real part greater than $-1$.
Since only one factor is differentiated in each term and the
remaining factors are locally bounded, the differentiated terms
are locally integrable, including at simultaneous collisions.

At $\xi_i=0$ and $\eta_j=0$, equivalently $u_i=\infty$ and
$v_j=\infty$, the exponential factors $e^{-\gamma_+c_1/\xi_i}$
and $e^{-\gamma_+c_1/\eta_j}$ dominate every fixed inverse power.
The boundary terms therefore vanish and the differentiated
densities are integrable there.
The separation of the two contour families in
\eqref{eq:screening-original-contours} excludes mixed collisions
away from their common origin, where the same exponential decay
controls the inverse powers introduced by differentiation.

Hence the boundary integrals vanish as the truncated domains increase
to the full integration domain, and the differentiated densities are
absolutely integrable.
Stokes' theorem on the compact truncations therefore implies that
the sum of the total-derivative terms has zero integral.

\textup{(ii)} We next account for the motion of the integration contours
when differentiating with respect to $z$, and show that the resulting
additional terms have zero integral. By \eqref{eq:screening-variable-change},
\[
 \left.\partial_z\xi_i\right|_x
 =(1-zx_i)^{-2}=\frac{\xi_i^2}{z^2}.
\]
The chain rule, including the Jacobian, gives
\begin{equation*}
 \partial_z\widetilde{\mathcal D}
 =\mathcal J\left[
 (\partial_z\mathcal D)_{\xi,\eta}
 +\sum_{i=1}^{m}
 \partial_{\xi_i}\left(\frac{\xi_i^2}{z^2}\mathcal D\right)
 \right].
\end{equation*}
The added divergences are the $k=1$ total-derivative terms
multiplied by the locally bounded factor $z^{-2}$,
and therefore integrate to zero by (i).
Lemma~\ref{lem:contour-differentiation} consequently identifies the
derivative of the moving integral with the integral of
$(\partial_z\mathcal D)_{\xi,\eta}$.
After integrating the commutator \eqref{eq:screening-current-commutator},
all total-derivative terms vanish by (i). The primary-field relation
\eqref{eq:free-field-vertex-commutator} and the source eigenvalue relations
\eqref{eq:coherent-virasoro-eigenvalues} give
\eqref{eq:screening-finite-ward-identity}.

\textup{(iii)} We apply (i) to finite coherent descendants and extend
its boundary argument to $-J\le k\le0$. Write
\[
 u=P(a_{-1},a_{-2},\ldots)|p_0;c_1\rangle,
\]
where $P$ is a polynomial. For $s>0$, we have
\[
 L_kS_+(\eta)a_{-s}-S_+(\eta)L_ka_{-s}
 =\partial_\eta\!\left[
 \eta^{k+1}(a_{-s}-\gamma_+\eta^{-s})S_+(\eta)
 \right].
\]
Applying this identity repeatedly shows that the factors contributed
by $P$ remain inside the total derivatives.
Under the substitution $\eta_j=1/y_j$, the factors contributed by $P$
become polynomials in $y_1,\ldots,y_n$ and the creation operators.
Since these polynomials contain no negative powers of the $y_j$,
they do not increase the total inverse $y$-degree, which remains
at most $\ell$ at output cutoff $\ell$.
Exponential decay controls the polynomial growth of these factors
at infinity. Thus the boundary estimates
in (i) remain valid for $1\le k\le K$.
For $-J\le k\le0$, the factor $y_j^{1-k}$ improves the zero-endpoint
behavior and introduces only polynomial growth at infinity, so the
same estimates apply. The boundary argument in (i) therefore gives
zero integral for these total-derivative terms throughout
$-J\le k\le K$, proving
\eqref{eq:fixed-screening-intertwining}.
\end{proof}

\subsection{Asymptotic expansion of the rank-one irregular conformal block}
\label{subsec:screening-asymptotic-expansion}

Let $\langle p_{\mathrm{out}}|$ be the Fock vacuum functional,
with $\langle p_{\mathrm{out}}|p_{\mathrm{out}}\rangle=1$ and
$\langle p_{\mathrm{out}}|a_{-s}=0$ for $s\ge1$.
Define the normalization constant $C_{m,n}$ and the vector
$|\Omega_n\rangle$ by the following limits. We will show that
$|\Omega_n\rangle$ is a normalized screened rank-one irregular
vector with parameters $\Lambda'$.
\begin{equation*}
 \begin{aligned}
 C_{m,n}&=\lim_{z\to0}z^{-\kappa}e^{-\beta/z}
                  \langle p_{\mathrm{out}}|\mathcal S_{m,n}(z)\rangle,\\
 \pi_{\le\ell}|\Omega_n\rangle
 &=C_{m,n}^{-1}\lim_{z\to0}z^{-\kappa}e^{-\beta/z}
                         \pi_{\le\ell}\mathcal S_{m,n}(z).
 \end{aligned}
\end{equation*}
The limits are taken in the chosen sector, in the convergence
domain for each fixed $\ell$.
Write $W(x,y)=W_x(x)W_y(y)$ for the two factors in
\eqref{eq:screening-weight}, with
\[
 Z_x=\int_{\Gamma_x^{(m)}(\varphi)}W_x(x)\,d^mx,
 \qquad
 Z_y=\int_{\Gamma_y^{(n)}(\psi)}W_y(y)\,d^ny.
\]
Applying Lemma~\ref{lem:contour-remainder} with $M=0$, we may pass
the limit $z\to0$ under the integral sign at each finite Fock
cutoff. Since $G(0;x,y)=1$, $\mathcal E(0;x,y)=\mathcal E_0(y)$,
and $W(x,y)=W_x(x)W_y(y)$, we obtain
\begin{equation*}
 C_{m,n}=\varepsilon_{m,n}Z_xZ_y,
 \qquad
 |\Omega_n\rangle=\frac1{Z_y}
       \int_{\Gamma_y^{(n)}(\psi)}W_y(y)\mathcal E_0(y)\,d^ny.
\end{equation*}
After parametrizing the contours by positive real variables, the
Laguerre--Selberg evaluation
\eqref{eq:Laguerre-Selberg-normalization} in
Subsection~\ref{subsec:screening-kernel-specialization} expresses
$Z_x$ and $Z_y$ as nonzero power factors times products and quotients
of Gamma functions. Under \eqref{eq:screening-convergence-domain},
all Gamma-function arguments have positive real part, so these
factors are finite and nonzero. Hence $Z_xZ_y\ne0$.

The vector $|\Omega_n\rangle$ is proportional to
$Q_{+,\psi}^{(n)}|p_{\mathrm{out}}-n\gamma_+;c_1\rangle$.
By the fixed-screening assertion of
Lemma~\ref{lem:contour-boundary-vanishing}, $|\Omega_n\rangle$
satisfies the rank-one irregular conditions with parameters
\[
 \Lambda'
 =\left((p_{\mathrm{out}}-n\gamma_+-2\rho)c_1,\frac{c_1^2}{2}\right)
 =(\Lambda_1-\beta,\Lambda_2).
\]
We henceforth interpret $|\Omega_n\rangle$ by coefficientwise meromorphic
continuation in the parameters.
It therefore defines the following realization of
$M_{\Lambda'}^{[1]}$ by a Virasoro homomorphism:
\begin{equation*}
 \jmath_n:M_{\Lambda'}^{[1]}\longrightarrow
                   \widehat{\mathcal F}_{p_{\mathrm{out}}},
 \qquad
 \jmath_n(U|\Lambda'\rangle)=U|\Omega_n\rangle,
 \quad U\in U(\Vir).
\end{equation*}
Thus $\jmath_n$ realizes the generating irregular vector
$|\Lambda'\rangle$ as the normalized screened rank-one irregular
vector $|\Omega_n\rangle$.
Since $\Lambda'_2=c_1^2/2\ne0$, the module $M_{\Lambda'}^{[1]}$
is irreducible \cite[Section~2.2]{NagoyaNakagawa2026}.
The normalization $\langle p_{\mathrm{out}}|\Omega_n\rangle=1$
shows that $\jmath_n$ is nonzero and hence injective. It therefore
identifies $M_{\Lambda'}^{[1]}$ with the cyclic submodule
$U(\Vir)|\Omega_n\rangle\subset\widehat{\mathcal F}_{p_{\mathrm{out}}}$.

Moreover, before the fixed screenings are applied, each creation
mode $a_{-s}$ multiplying the coherent state is accompanied by at
least $z^s$. Put
\[
 R_m(z)=\int_{\Gamma_\xi^{(m)}(z)}
 V_{\gamma_z}(z)S_+(\xi_m)\cdots S_+(\xi_1)
 |p;c_1\rangle\,d^m\xi.
\]
After removing the scalar prefactor from $R_m(z)$, the coefficient of $z^j$
in the resulting formal asymptotic expansion belongs to
$\mathcal F_{p_*,c_1,\le j}^{[1]}$, where
$p_*=p_{\mathrm{out}}-n\gamma_+$.
By \eqref{eq:irregular-fock-identification} and the fixed-screening
assertion of Lemma~\ref{lem:contour-boundary-vanishing}, the
asymptotic expansion supplied by Lemma~\ref{lem:contour-remainder}
therefore has the form
\[
 C_{m,n}^{-1}\mathcal S_{m,n}(z)
 \sim z^\kappa e^{\beta/z}
      \sum_{j\ge0}z^j\jmath_n(\widetilde v_j),
 \qquad \widetilde v_j\in U_j,\quad
 \widetilde v_0=|\Lambda'\rangle.
\]
The asymptotic relation is interpreted through order $N$ at Fock cutoff $\ell$.
The coefficients of the resulting formal expansion extend meromorphically to
generic parameters beyond the convergence domain \eqref{eq:screening-convergence-domain}.

\begin{theorem}[Vector asymptotic expansion]
\label{thm:screening-vector-identification}
Let $m,n,N,\ell\in\mathbb Z_{\ge0}$ and assume
\eqref{eq:screening-convergence-domain}.
Under the realization $\jmath_n$, the normalized screening integral
has the following formal asymptotic expansion:
\begin{equation}
 C_{m,n}^{-1}\mathcal S_{m,n}(z)
 \sim\jmath_n\bigl(\Phi_{\Lambda',\Lambda}^{\Delta(\gamma_z)}(z)
                         |\Lambda\rangle\bigr).
 \label{eq:screening-vector-formal-identification}
\end{equation}
\end{theorem}

\begin{proof}
Compare the finite-order asymptotic coefficients in
\eqref{eq:screening-finite-ward-identity}. Since $\jmath_n$ is an
injective Virasoro homomorphism, we obtain
\begin{equation*}
 \begin{aligned}
 &(L_k-\delta_{k,1}\Lambda_1-\delta_{k,2}\Lambda_2)\widetilde v_j\\
 &\qquad=-\beta\widetilde v_{j+1-k}
   +\bigl(\kappa+(k+1)\Delta(\gamma_z)+j-k\bigr)
       \widetilde v_{j-k},\qquad k\ge1,
 \end{aligned}
\end{equation*}
where $\widetilde v_a=0$ for $a<0$.
Together with $\widetilde v_0=|\Lambda'\rangle$, these are precisely
the compatibility conditions for the irregular vertex operator.
Let $v_j$ be defined by
\[
 \Phi_{\Lambda',\Lambda}^{\Delta(\gamma_z)}(z)|\Lambda\rangle
 =z^\kappa e^{\beta/z}\sum_{j\ge0}v_jz^j.
\]
The uniqueness in Theorem~\ref{thm:irregular-vertex-existence}
therefore gives $\widetilde v_j=v_j$ for all $j\ge0$.
\end{proof}

\begin{corollary}[Conformal block]
\label{cor:screening-block-identification}
Under the hypotheses of
Theorem~\ref{thm:screening-vector-identification} with $\ell=0$,
put $\Delta_\infty=\Delta(\gamma_\infty)$ and
\begin{equation*}
 \mathcal I_{m,n}(z)
 =C_{m,n}^{-1}\langle p_{\mathrm{out}}|
                         \mathcal S_{m,n}(z)\rangle.
\end{equation*}
Then
\begin{equation*}
 \mathcal I_{m,n}(z)
 =z^\kappa e^{\beta/z}
   \left(\sum_{j=0}^{N}b_jz^j+O(z^{N+1})\right),
\end{equation*}
where $b_j=\langle\Delta_\infty|v_j\rangle$.
Its formal expansion is
$\mathcal B(z;c,\Delta_\infty,\Delta(\gamma_z),\Lambda_1,\Lambda_2,\beta)$.
Equivalently,
\begin{equation}
 b_j=
 \frac{\displaystyle\int W(x,y)[z^j]G(z;x,y)\,d^mx\,d^ny}
      {\displaystyle\int W(x,y)\,d^mx\,d^ny},
 \qquad j\ge0.
 \label{eq:screening-block-coefficient-formula}
\end{equation}
\end{corollary}

\begin{proof}
The intertwining property of $\jmath_n$ and
$\Delta(p_{\mathrm{out}})=\Delta(\gamma_\infty)$ imply that
$v\mapsto\langle p_{\mathrm{out}}|\jmath_n(v)\rangle$ has weight
$\Delta_\infty$, annihilates $L_kv$ for $k<0$, and takes the
value one on $|\Lambda'\rangle$. By uniqueness of the pairing
\eqref{eq:rank-one-pairing}, this functional equals
$v\mapsto\langle\Delta_\infty|v\rangle$.
Apply this functional to
\eqref{eq:screening-vector-formal-identification}. Since the Fock vacuum
functional is the level-zero projection, the remainder has the
stated order. Finally,
$\langle p_{\mathrm{out}}|\mathcal E(z;x,y)\rangle=1$ gives
\eqref{eq:screening-block-coefficient-formula}.
\end{proof}

\section{Jack expansion}
\label{sec:jack-expansion}

Recall the scalar integration kernel in
\eqref{eq:screening-scalar-factor}:
\begin{equation}
 G(z;x,y)=
 \prod_{i=1}^{m}(1-zx_i)^{\gamma_+\gamma_\infty}
 \prod_{j=1}^{n}(1-zy_j)^{\gamma_+\gamma_z}
 \prod_{i=1}^{m}\prod_{j=1}^{n}
          \bigl(1-z(x_i+y_j)\bigr)^{\gamma_+^2}.
 \label{eq:jack-screening-kernel}
\end{equation}
Integrating $[z^j]G(z;x,y)$ against $W(x,y)$ and dividing by $Z_xZ_y$
gives the coefficient $b_j$ in \eqref{eq:screening-block-coefficient-formula}.
We express this kernel as an operator flow, compute its Jack coefficients,
and evaluate the resulting Laguerre--Jack integrals.
The symmetric-function expansions below are formal over
$\mathbb Q(\alpha)$, with the other parameters adjoined.

\subsection{Jack and shifted Jack polynomials}

A partition $\lambda=(\lambda_1,\lambda_2,\ldots)$ has size
$|\lambda|=\sum_i\lambda_i$, length $\ell(\lambda)$, and conjugate
$\lambda'$.  We identify it with the diagram
$\{(i,j):1\leq j\leq\lambda_i\}$.  For a cell $s=(i,j)\in\lambda$, set
\begin{equation*}
 \begin{aligned}
 a(s)&=\lambda_i-j,\quad \ell(s)=\lambda_j'-i,\quad
 a'(s)=j-1,\quad \ell'(s)=i-1,\\[4pt]
 \operatorname{ct}_\theta(s)&=a'(s)-\theta\ell'(s).
 \end{aligned}
\end{equation*}
Containment $\nu\subseteq\lambda$ means $\nu_i\leq\lambda_i$ for every
$i$. We use Macdonald's Jack parameter $\alpha$ and put
\begin{equation*}
 \theta=\alpha^{-1}.
\end{equation*}

For $\varrho=(1^{m_1}2^{m_2}\cdots)$, write
$z_\varrho=\prod_{r\geq1}r^{m_r}m_r!$.  The Jack scalar product is
\begin{equation}
 \langle p_\lambda,p_\mu\rangle_{\alpha}
 =\delta_{\lambda\mu}z_\lambda\alpha^{\ell(\lambda)}.
 \label{eq:jack-scalar-product}
\end{equation}
Let $P_\lambda^{(\alpha)}$ be the monic Jack polynomial, characterized by
triangularity in the monomial basis and orthogonality for
\eqref{eq:jack-scalar-product}; see
\cite[Chapter~VI, Section~10, equations~(10.13)--(10.14)]{Macdonald1995}.
Define
\begin{equation}
 \begin{aligned}
 H_\theta(\lambda)
 &=\prod_{s\in\lambda}\bigl(a(s)+\theta\ell(s)+1\bigr),\quad
 H_\theta'(\lambda)
 =\prod_{s\in\lambda}\bigl(a(s)+\theta\ell(s)+\theta\bigr),\\[4pt]
 b_\lambda(\alpha)&=\frac{H_\theta'(\lambda)}{H_\theta(\lambda)}.
 \end{aligned}
 \label{eq:hook-products}
\end{equation}
The dual Jack polynomial and the norm of the monic polynomial are therefore
\begin{equation}
 Q_\lambda^{(\alpha)}=b_\lambda(\alpha)P_\lambda^{(\alpha)},
 \qquad
 \langle P_\lambda^{(\alpha)},P_\mu^{(\alpha)}\rangle_{\alpha}
 =\delta_{\lambda\mu}\frac{H_\theta(\lambda)}{H_\theta'(\lambda)}
 =\frac{\delta_{\lambda\mu}}{b_\lambda(\alpha)}.
 \label{eq:PQ-normalization}
\end{equation}
These products are related to Macdonald's hook products by
$H_\theta(\lambda)=\alpha^{-|\lambda|}c_\lambda'(\alpha)$ and
$H_\theta'(\lambda)=\alpha^{-|\lambda|}c_\lambda(\alpha)$.
For the duality, norm, and hook formulas, see
\cite[Chapter~VI, Section~10, equations~(10.16) and~(10.21)]{Macdonald1995}.

We use the shifted Jack polynomials of Okounkov and Olshanski
\cite[Section~2, equation~(2.2) and the preceding definition]{OkounkovOlshanski1997}, normalized by
\begin{equation*}
 P_\nu^*(\nu;\theta)=H_\theta(\nu),
 \qquad
 P_\nu^*(\lambda;\theta)=0\quad\text{if }\nu\not\subseteq\lambda.
\end{equation*}
For $\nu\subseteq\lambda$, define the generalized skew Pochhammer symbol by
\begin{equation}
 (a)_{\lambda/\nu}^{(\theta)}
 =\prod_{s\in\lambda\setminus\nu}
   \bigl(a+a'(s)-\theta\ell'(s)\bigr)
 =\prod_{s\in\lambda\setminus\nu}
   \bigl(a+\operatorname{ct}_\theta(s)\bigr).
 \label{eq:skew-Pochhammer}
\end{equation}
The non-skew notation $(a)_\lambda^{(\theta)}$ means
$(a)_{\lambda/\varnothing}^{(\theta)}$.  The coefficient that will occur in the
exponential operator is
\begin{equation*}
 {
 A_{\lambda,\nu}^{(\alpha)}(a)
 =\frac{P_\nu^*(\lambda;\theta)}{H_\theta(\lambda)}
  (a)_{\lambda/\nu}^{(\theta)}.}
\end{equation*}
We set $A_{\lambda,\nu}^{(\alpha)}(a)=0$ when $\nu\not\subseteq\lambda$.

\subsection{Operator flow and explicit Jack coefficients}
\label{subsec:Fock-flow}
\label{subsec:exponential-Jack-coefficients}

Let $\mathcal S=\mathbb Q(\alpha)[p_1,p_2,\ldots]$.  Write
$\mathcal S=\bigoplus_{d\geq0}\mathcal S_d$ and
$\widehat{\mathcal S}=\prod_{d\geq0}\mathcal S_d$, where $\mathcal S_d$ is the
homogeneous degree-$d$ subspace.
We adjoin any additional parameters to the coefficient field. The symbol $\dagger$ denotes the bilinear
Jack adjoint, interpreted coefficientwise with $z$ unchanged. Define
\begin{equation}
 \mathsf E=\sum_{r\geq1}r p_r\frac{\partial}{\partial p_r},
 \qquad
 \mathsf D_a=\sum_{r\geq1}r p_{r+1}\frac{\partial}{\partial p_r}+ap_1,
 \qquad
 {\mathsf T_a(z)=e^{z\mathsf D_a}z^{\mathsf E}}.
 \label{eq:Fock-triangular-operators}
\end{equation}
On a finite alphabet, the derivation
\begin{equation*}
 \sum_{r\geq1}r p_{r+1}\frac{\partial}{\partial p_r}
 =\sum_i x_i^2\frac{\partial}{\partial x_i}
\end{equation*}
has flow $x_i\mapsto x_i/(1-zx_i)$.  Including the multiplication term $ap_1$ gives
\begin{equation*}
 e^{z\mathsf D_a}f(x)
 =\prod_i(1-zx_i)^{-a}
 f\!\left(\frac{x_1}{1-zx_1},
           \frac{x_2}{1-zx_2},\ldots\right).
\end{equation*}
Indeed, the right-hand side $F(z,x)$ satisfies
$\partial_zF=\mathsf D_aF$ and $F(0,x)=f(x)$; comparing powers of $z$
gives the unique formal solution $F=e^{z\mathsf D_a}f$.
Since $z^{\mathsf E}f(x)=f(zx)$, it follows that
\begin{equation}
 \mathsf T_a(z)f(x)
 =\prod_i(1-zx_i)^{-a}
 f\!\left(\frac{zx_1}{1-zx_1},
           \frac{zx_2}{1-zx_2},\ldots\right).
 \label{eq:T-flow}
\end{equation}
Equivalently, $[\mathsf E,\mathsf D_a]=\mathsf D_a$ gives
\begin{equation*}
 e^{z\mathsf D_a}z^{\mathsf E}=z^{\mathsf E}e^{\mathsf D_a}.
\end{equation*}
Here $e^{\mathsf D_a}$ is defined on $\widehat{\mathcal S}$ degree by degree.
For $f=\sum_{q\ge0}f_q$, with $f_q\in\mathcal S_q$, one has
\[
 [z^d]\mathsf T_a(z)f=\sum_{q=0}^{d}
 \frac{\mathsf D_a^{d-q}}{(d-q)!}f_q\in\mathcal S_d.
\]

Define the oscillator-number operator and the mixed-kernel operator by
\begin{equation}
 N_{\mathrm{osc}}p_\varrho=\ell(\varrho)p_\varrho,
 \qquad
 {\mathsf K(w)=(\alpha w)^{N_{\mathrm{osc}}}},
 \qquad
 \mathsf K(w)p_\varrho=(\alpha w)^{\ell(\varrho)}p_\varrho.
 \label{eq:K-operator-definition}
\end{equation}

Define the mixed kernel by
\begin{equation}
 \mathcal K_w(x,y)
 =\exp\!\left(w\sum_{r\ge1}\frac{p_r(x)p_r(y)}{r}\right)
 =\prod_{i,j}(1-x_i y_j)^{-w}.
 \label{eq:flow-mixed-kernel}
\end{equation}
For $x=(x_1,\ldots,x_m)$ and $y=(y_1,\ldots,y_n)$, define
\begin{equation}
 F_{m,n}(x,y;z;u,v,w)
 :=\prod_{i=1}^{m}(1-zx_i)^{-u}
 \prod_{j=1}^{n}(1-zy_j)^{-v}
 \prod_{i=1}^{m}\prod_{j=1}^{n}(1-zx_i-zy_j)^{-w}.
 \label{eq:universal-two-alphabet-function}
\end{equation}
Put $a_{\mathrm L}=u+nw$ and $a_{\mathrm R}=v+mw$. Superscripts
$(x)$ and $(y)$ indicate the alphabet on which an operator acts.
The flow formula \eqref{eq:T-flow} gives
\begin{equation}
 F_{m,n}
 =\mathsf T_{a_{\mathrm L}}^{(x)}(z)
   \mathsf T_{a_{\mathrm R}}^{(y)}(z)\mathcal K_w(x,y).
 \label{eq:universal-kernel-flow}
\end{equation}
To verify this directly, set
\begin{equation*}
 X_i=\frac{zx_i}{1-zx_i},\qquad
 Y_j=\frac{zy_j}{1-zy_j}.
\end{equation*}
Then
\begin{equation*}
 1-zx_i-zy_j=(1-zx_i)(1-zy_j)(1-X_iY_j),
\end{equation*}
which proves \eqref{eq:universal-kernel-flow} by \eqref{eq:T-flow}.
Only mixed-kernel components of bidegree $(r,r)$ with $2r\le d$
contribute to $[z^d]F_{m,n}$, so the action is well defined coefficientwise.

In particular, the screening kernel \eqref{eq:jack-screening-kernel} is
\begin{equation*}
 G(z;x,y)=
 \mathsf T_{-\gamma_+\gamma_\infty-n\gamma_+^2}^{(x)}(z)
 \mathsf T_{-\gamma_+\gamma_z-m\gamma_+^2}^{(y)}(z)
 \mathcal K_{-\gamma_+^2}(x,y).
\end{equation*}

Write $[P_\lambda]f$ for the coefficient of $P_\lambda^{(\alpha)}$ in $f$.

\begin{theorem}[Explicit coefficients of the exponential operator]
\label{thm:exponential-Jack-coefficients}
For every partition $\nu$, one has, in the degree completion,
\begin{equation}
 e^{\mathsf D_a}P_\nu^{(\alpha)}
 =\sum_{\lambda\supseteq\nu}
 \frac{P_\nu^*(\lambda;\theta)}{H_\theta(\lambda)}
 (a)_{\lambda/\nu}^{(\theta)}P_\lambda^{(\alpha)}.
 \label{eq:unit-exponential-Jack-expansion}
\end{equation}
Consequently,
\begin{equation}
 [P_\lambda]\mathsf T_a(z)P_\nu
 =z^{|\lambda|}A_{\lambda,\nu}^{(\alpha)}(a).
 \label{eq:T-P-matrix}
\end{equation}
\end{theorem}

\begin{proof}
Write $P_\mu=P_\mu^{(\alpha)}$ throughout the proof.  Consider the
degree-preserving Laplace--Beltrami operator
of \cite[Chapter~VI, Section~3, Example~3(e), and Section~4,
Example~3]{Macdonald1995}, in the normalization obtained by dividing
Macdonald's operator by $\alpha$:
\begin{align}
 \mathsf H_\theta={}&
 \frac12\sum_{r,s\geq1}rs\,p_{r+s}
       \frac{\partial^2}{\partial p_r\partial p_s}
 +\frac{\theta}{2}\sum_{r,s\geq1}(r+s)p_rp_s
       \frac{\partial}{\partial p_{r+s}}
 \notag\\
 &+\frac{1-\theta}{2}\sum_{r\geq1}r(r-1)p_r
       \frac{\partial}{\partial p_r}.
 \label{eq:content-Laplace-Beltrami}
\end{align}
Its eigenvalue on $P_\mu$ is given by
\begin{equation}
 \mathsf H_\theta P_\mu
 =\varepsilon_\theta(\mu)P_\mu,
 \qquad
 \varepsilon_\theta(\mu)
 =\sum_i\left\{\binom{\mu_i}{2}-\theta(i-1)\mu_i\right\}
 =\sum_{s\in\mu}\operatorname{ct}_\theta(s).
 \label{eq:content-Laplace-Beltrami-eigenvalue}
\end{equation}
A direct calculation from \eqref{eq:content-Laplace-Beltrami} gives
\begin{equation}
 [\mathsf H_\theta,p_1]
 =\sum_{r\geq1}r p_{r+1}\frac{\partial}{\partial p_r},
 \qquad
 \mathsf D_a=[\mathsf H_\theta,p_1]+ap_1,
 \label{eq:Da-content-commutator}
\end{equation}
where $p_1$ also denotes multiplication by $p_1$.

Write the one-box Pieri rule as
\begin{equation*}
 p_1P_\mu
 =\sum_{\lambda=\mu+\square}
   \pi_{\lambda\mu}(\theta)P_\lambda;
\end{equation*}
see \cite[Chapter~VI, (6.24)(iv), with the Jack specialization of
Section~10]{Macdonald1995}.
Equations \eqref{eq:content-Laplace-Beltrami-eigenvalue} and
\eqref{eq:Da-content-commutator} imply
\begin{equation}
 \mathsf D_aP_\mu
 =\sum_{\lambda=\mu+\square}
   \pi_{\lambda\mu}(\theta)
   \bigl(a+\operatorname{ct}_\theta(\square)\bigr)P_\lambda.
 \label{eq:Da-weighted-Pieri}
\end{equation}
For $k=|\lambda|-|\nu|\geq0$, let
\begin{equation*}
 d_\theta(\lambda/\nu)=[P_\lambda]p_1^kP_\nu.
\end{equation*}
Iterating \eqref{eq:Da-weighted-Pieri} gives a sum over chains
$\nu=\mu^{(0)}\nearrow\mu^{(1)}\nearrow\cdots
\nearrow\mu^{(k)}=\lambda$ of one-box additions.  The product of the
content factors is $(a)_{\lambda/\nu}^{(\theta)}$ on every such chain,
since the content of a box depends only on its coordinates.  The
remaining sum of products of Pieri coefficients is
$d_\theta(\lambda/\nu)$.  Consequently,
\begin{equation}
 [P_\lambda]\mathsf D_a^kP_\nu
 =(a)_{\lambda/\nu}^{(\theta)}d_\theta(\lambda/\nu).
 \label{eq:Da-chain-factorization}
\end{equation}

Okounkov and Olshanski
\cite[Section~5, equations~(5.1) and~(5.5)]{OkounkovOlshanski1997}
express $d_\theta(\lambda/\nu)$ in terms of shifted Jack polynomials:
\begin{equation}
 \frac{d_\theta(\lambda/\nu)}{k!}
 =\frac{P_\nu^*(\lambda;\theta)}{H_\theta(\lambda)}.
 \label{eq:shifted-Pieri-coefficient-evaluation}
\end{equation}

Dividing \eqref{eq:Da-chain-factorization} by $k!$ and using
\eqref{eq:shifted-Pieri-coefficient-evaluation} gives the degree
$|\nu|+k$ component of \eqref{eq:unit-exponential-Jack-expansion}.
The matrix formula \eqref{eq:T-P-matrix} follows from
$\mathsf T_a(z)=z^{\mathsf E}e^{\mathsf D_a}$.
\end{proof}

Combining Theorem~\ref{thm:exponential-Jack-coefficients} with the flow
formula \eqref{eq:T-flow} recovers the fractional-linear Jack formula of
Lascoux and Warnaar \cite[Corollary~5.3]{LascouxWarnaar2011} in the monic
normalization: for $x=(x_1,\ldots,x_m)$ and $\ell(\nu)\le m$,
\begin{equation*}
 \prod_{i=1}^{m}(1-zx_i)^{-a}
 P_\nu^{(\alpha)}\!\left(
 \frac{zx_1}{1-zx_1},\ldots,\frac{zx_m}{1-zx_m}\right)
 =\sum_{\substack{\lambda\supseteq\nu\\\ell(\lambda)\le m}}
 z^{|\lambda|}A_{\lambda,\nu}^{(\alpha)}(a)P_\lambda^{(\alpha)}(x).
\end{equation*}

\subsection{Universal Jack expansion and operator factorization}

Define the power-sum-to-$P$ transition coefficients by
\begin{equation}
 p_\varrho=
 \sum_{\nu\vdash|\varrho|}
 \Theta_{\varrho,\nu}^{(\alpha)}P_\nu^{(\alpha)}.
 \label{eq:Theta-definition}
\end{equation}
At $\alpha=1$, these coefficients are the ordinary symmetric-group
characters $\Theta_{\varrho,\nu}^{(1)}=\chi^\nu(\varrho)$.
For general $\alpha$, the map formula of Ben Dali and
{\fontencoding{T1}\selectfont Do\l{}\k{e}ga}
\cite[Theorem~1.4]{BenDaliDolega2023} evaluates the power-sum coefficients
of the integral-form Jack polynomial
$J_\nu^{(\alpha)}=\alpha^{|\nu|}H'_\theta(\nu)P_\nu^{(\alpha)}$.
In our normalization, Jack orthogonality gives
\[
 \Theta_{\varrho,\nu}^{(\alpha)}
 =\frac{z_\varrho\alpha^{\ell(\varrho)-|\nu|}}{H_\theta(\nu)}
 [p_\varrho]J_\nu^{(\alpha)}.
\]
For $|\nu|=|\eta|=r$, set
\begin{equation}
 {
 K_{\nu,\eta}^{(\alpha)}(w)
 =\sum_{\varrho\vdash r}\frac{w^{\ell(\varrho)}}{z_\varrho}
 \Theta_{\varrho,\nu}^{(\alpha)}
 \Theta_{\varrho,\eta}^{(\alpha)}.}
 \label{eq:K-definition}
\end{equation}
Set $K_{\nu,\eta}^{(\alpha)}(w)=0$ if $|\nu|\ne|\eta|$.

\begin{lemma}[Universal mixed Jack kernel]
\label{lem:mixed-Jack-kernel}
The mixed kernel \eqref{eq:flow-mixed-kernel} has expansion
\begin{equation*}
 \mathcal K_w(x,y)
 =\sum_{|\nu|=|\eta|}
 K_{\nu,\eta}^{(\alpha)}(w)
 P_\nu^{(\alpha)}(x)P_\eta^{(\alpha)}(y).
\end{equation*}
\end{lemma}

\begin{proof}
Expanding the exponential in \eqref{eq:flow-mixed-kernel} gives
\[
 \mathcal K_w(x,y)
 =\sum_\varrho\frac{w^{\ell(\varrho)}}{z_\varrho}
 p_\varrho(x)p_\varrho(y).
\]
Substitute \eqref{eq:Theta-definition} in both alphabets.
\end{proof}

\begin{theorem}[Universal Jack expansion and operator factorization]
\label{thm:universal-jack}
\label{thm:universal-sandwich}
For $m,n\in\mathbb Z_{\ge0}$ and independent parameters $u,v,w$, put
\begin{equation}
 \mathsf C_{m,n}(z;u,v,w)
 =\mathsf T_{u+nw}(z)\mathsf K(w)
  \mathsf T_{v+mw}(z)^\dagger.
 \label{eq:universal-sandwich}
\end{equation}
Define the coefficients by the finite sums
\begin{equation}
 C_{\lambda,\mu}^{(\alpha)}(u,v,w;m,n)
 =\sum_{\substack{\nu\subseteq\lambda,\ \eta\subseteq\mu\\
                   |\nu|=|\eta|}}
 A_{\lambda,\nu}^{(\alpha)}(u+nw)
 K_{\nu,\eta}^{(\alpha)}(w)
 A_{\mu,\eta}^{(\alpha)}(v+mw).
 \label{eq:universal-Jack-coefficient}
\end{equation}
As a formal series in $z$, the flow kernel \eqref{eq:universal-two-alphabet-function} has expansion
\begin{equation}
 F_{m,n}
 =\sum_{\substack{\ell(\lambda)\le m\\\ell(\mu)\le n}}
 z^{|\lambda|+|\mu|}
 C_{\lambda,\mu}^{(\alpha)}(u,v,w;m,n)
 P_\lambda^{(\alpha)}(x)P_\mu^{(\alpha)}(y).
 \label{eq:universal-Jack-expansion}
\end{equation}
The same coefficients give the matrix of the operator
\eqref{eq:universal-sandwich} on the Jack Fock space:
\begin{equation}
 [P_\lambda]\mathsf C_{m,n}(z;u,v,w)P_\mu
 =z^{|\lambda|+|\mu|}
 \frac{C_{\lambda,\mu}^{(\alpha)}(u,v,w;m,n)}{b_\mu(\alpha)}.
 \label{eq:universal-sandwich-matrix}
\end{equation}
The matrix formula holds for all partitions $\lambda,\mu$.
\end{theorem}

\begin{proof}
Expand $\mathcal K_w$ by Lemma~\ref{lem:mixed-Jack-kernel} and apply
\eqref{eq:T-P-matrix} in both alphabets of
\eqref{eq:universal-kernel-flow}. This gives
\eqref{eq:universal-Jack-coefficient} and
\eqref{eq:universal-Jack-expansion}. The length restrictions follow by
specializing symmetric functions to $m$ and $n$ variables.

By \eqref{eq:PQ-normalization}, the adjoint matrix is
\begin{equation}
 [P_\eta]\mathsf T_a(z)^\dagger P_\mu
 =z^{|\mu|}A_{\mu,\eta}^{(\alpha)}(a)
   \frac{b_\eta(\alpha)}{b_\mu(\alpha)}.
 \label{eq:T-adjoint-P-matrix}
\end{equation}
Pairing \eqref{eq:Theta-definition} with $P_\eta$ gives
\[
 [p_\varrho]P_\eta
 =\frac{\Theta_{\varrho,\eta}^{(\alpha)}}
 {z_\varrho\alpha^{\ell(\varrho)}b_\eta(\alpha)}.
\]
Hence \eqref{eq:K-operator-definition} yields, for $|\nu|=|\eta|$,
\begin{equation}
 [P_\nu]\mathsf K(w)P_\eta
 =\frac{1}{b_\eta(\alpha)}
 \sum_{\varrho\vdash|\eta|}\frac{w^{\ell(\varrho)}}{z_\varrho}
 \Theta_{\varrho,\nu}^{(\alpha)}\Theta_{\varrho,\eta}^{(\alpha)}
 =\frac{K_{\nu,\eta}^{(\alpha)}(w)}{b_\eta(\alpha)}.
 \label{eq:K-P-matrix}
\end{equation}
Both sides vanish when $|\nu|\ne|\eta|$, since $\mathsf K(w)$ preserves degree.
Multiplying the matrices given by \eqref{eq:T-P-matrix}, \eqref{eq:K-P-matrix}, and
\eqref{eq:T-adjoint-P-matrix} gives \eqref{eq:universal-sandwich-matrix}:
the intermediate factors $b_\eta^{-1}$ and $b_\eta/b_\mu$ cancel to $b_\mu^{-1}$,
and the remaining sum is \eqref{eq:universal-Jack-coefficient}.
The adjoint coefficient $[z^d]\mathsf T_a(z)^\dagger$ depends only on the
degree-$d$ input component and has output degrees at most $d$.
In the product, $[z^N]$ involves only $|\lambda|+|\mu|=N$,
so it is well defined on $\widehat{\mathcal S}$ coefficientwise.
\end{proof}

\subsection{Application to the rank-one irregular conformal block}
\label{subsec:screening-kernel-specialization}
\label{subsec:general-coefficient-formula}

The Vandermonde exponent in \eqref{eq:screening-weight} fixes the
Jack parameter to
\begin{equation}
 \alpha=\frac{2}{\gamma_+^2},
 \qquad \theta=\frac{\gamma_+^2}{2}.
 \label{eq:screening-Jack-parameter}
\end{equation}
The kernel \eqref{eq:jack-screening-kernel} is the specialization
of \eqref{eq:universal-two-alphabet-function} given by
\begin{equation}
 u=-\gamma_+\gamma_\infty,\qquad
 v=-\gamma_+\gamma_z,\qquad w=-\gamma_+^2.
 \label{eq:screening-universal-specialization}
\end{equation}
At this specialization, $\alpha w=-2$, so the middle operator is
$\mathsf K(-\gamma_+^2)=(-2)^{N_{\mathrm{osc}}}$.
Consequently, Theorem~\ref{thm:universal-jack} gives
\begin{equation}
 G(z;x,y)=
 \sum_{\substack{\ell(\lambda)\le m\\\ell(\mu)\le n}}
 z^{|\lambda|+|\mu|}
 C_{\lambda,\mu}^{(\alpha)}
       (-\gamma_+\gamma_\infty,-\gamma_+\gamma_z,-\gamma_+^2;m,n)
 P_\lambda^{(\alpha)}(x)
 P_\mu^{(\alpha)}(y).
 \label{eq:screening-Jack-expansion}
\end{equation}
We now evaluate the integrals of these Jack polynomials against the weight
$W(x,y)=W_x(x)W_y(y)$.

We recall the Laguerre limit of the Selberg integral
\cite[equation~(1.1)]{ForresterWarnaar2008}:
\begin{equation}
 \begin{aligned}
 Z_M(a,\theta;q)
 &:=
 \int_{0<t_1<\cdots<t_M<\infty}
 \prod_{i<j}(t_j-t_i)^{2\theta}
 \prod_{i=1}^{M}t_i^a e^{-qt_i}\,d^Mt\\
 &=
 \frac{q^{-M(a+1)-\theta M(M-1)}}{M!}
 \prod_{j=1}^{M}
 \frac{\Gamma(1+j\theta)\Gamma(a+1+(j-1)\theta)}
      {\Gamma(1+\theta)}.
 \end{aligned}
 \label{eq:Laguerre-Selberg-normalization}
\end{equation}
This integral is absolutely convergent and nonzero for
$\operatorname{Re}a>-1$, $\operatorname{Re}\theta\ge0$,
and $\operatorname{Re}q>0$, with $Z_0=1$.
The powers of the integration variables are defined with real logarithms,
and the branch of $\log q$ is real for $q>0$.
The factor $1/M!$ accounts for the ordered chamber.

The Laguerre limit of Kadell's Selberg--Jack formula
\cite{Kadell1997} (see also \cite[Theorem~3]{Warnaar2007}),
followed by scaling, gives
\[
 \begin{aligned}
 &\frac{1}{Z_M(a,\theta;q)}
 \int_{0<t_1<\cdots<t_M<\infty}
 P_\lambda^{(\alpha)}(t)
 \prod_{i<j}(t_j-t_i)^{2\theta}
 \prod_{i=1}^{M}t_i^a e^{-qt_i}\,d^Mt\\
 &\qquad=
 q^{-|\lambda|}P_\lambda^{(\alpha)}(1^M)
 \bigl(a+1+(M-1)\theta\bigr)_\lambda^{(\theta)},
 \end{aligned}
\]
where $\alpha=\theta^{-1}$, $\theta\ne0$ is generic,
and $\ell(\lambda)\le M$.

The same normalized formula holds on the ordered ray
$t_i=e^{\mathrm i\chi}u_i$, $0<u_1<\cdots<u_M<\infty$, provided
$\operatorname{Re}(qe^{\mathrm i\chi})>0$, with logarithms continued
from the positive ray. Indeed, the common weight and Jacobian
factors cancel, while homogeneity gives
$e^{\mathrm i\chi|\lambda|}(qe^{\mathrm i\chi})^{-|\lambda|}=q^{-|\lambda|}$.
Applying the normalized formula to $W_x$ and $W_y$, with
Laplace coefficients $-\gamma_+c_1$ and $\gamma_+c_1$,
respectively, gives
\begin{align}
 \frac{\displaystyle\int_{\Gamma_x^{(m)}(\varphi)}
 P_\lambda^{(\alpha)}(x)W_x(x)\,\mathrm d^mx}{Z_x}
 &=(-\gamma_+c_1)^{-|\lambda|}
   \frac{(m\theta)_\lambda^{(\theta)}
         (1+\gamma_+\gamma_z+(m-1)\theta)_\lambda^{(\theta)}}{H'_\theta(\lambda)},
 \label{eq:screening-x-Jack-average}\\
 \frac{\displaystyle\int_{\Gamma_y^{(n)}(\psi)}
 P_\mu^{(\alpha)}(y)W_y(y)\,\mathrm d^ny}{Z_y}
 &=(\gamma_+c_1)^{-|\mu|}
   \frac{(n\theta)_\mu^{(\theta)}
         (1+\gamma_+\gamma_\infty+(n-1)\theta)_\mu^{(\theta)}}{H'_\theta(\mu)}.
 \label{eq:screening-y-Jack-average}
\end{align}
These evaluations extend meromorphically in the parameters. Here we have substituted
$P_\lambda^{(\alpha)}(1^M)
=(M\theta)_\lambda^{(\theta)}/H'_\theta(\lambda)$.

We now remove the integer screening restriction.
Take $p,p_{\mathrm t},\gamma_z,\gamma_\infty\in\mathbb C$ and
$c_1\ne0$, and put $\delta=p_{\mathrm t}-p$.
The source parameters $\Lambda$ are given by
\eqref{eq:rank-one-free-field-parameters}; the target parameters
$\Lambda'$ are obtained by replacing $p$ with $p_{\mathrm t}$.
The prefactor parameters of \eqref{eq:rank-one-parameters} become
\begin{equation}
 \beta=-c_1\delta,\qquad
 \kappa=\delta(p_{\mathrm t}-2\rho)-2\Delta(\gamma_z).
 \label{eq:generic-prefactor-parameters}
\end{equation}
Substituting the screening conditions $p_{\mathrm t}=p+\gamma_z+m\gamma_+$
and $2\rho-p_{\mathrm t}-\gamma_\infty=n\gamma_+$ into
\eqref{eq:screening-x-Jack-average} and
\eqref{eq:screening-y-Jack-average} expresses the content parameters
in terms of charges.
For general charges, we use the same Laplace scales
$-\gamma_+c_1$ and $\gamma_+c_1$, and define
\begin{equation}
 \begin{aligned}
 x_-&=\frac{\gamma_+}{2}(\delta-\gamma_z),&
 x_+&=\frac{\gamma_+}{2}(\delta+\gamma_z-2\rho),\\
 y_-&=\frac{\gamma_+}{2}(2\rho-p_{\mathrm t}-\gamma_\infty),&
 y_+&=\frac{\gamma_+}{2}(\gamma_\infty-p_{\mathrm t}),
 \end{aligned}
 \label{eq:generic-cap-parameters}
\end{equation}
Rewriting $u+nw$ and $v+mw$ under the specialization
\eqref{eq:screening-universal-specialization} and the screening
conditions above gives the following definitions for general charges:
\begin{equation}
 a_{\mathrm L}=\gamma_+(p_{\mathrm t}-2\rho),\qquad
 a_{\mathrm R}=-\gamma_+\delta.
 \label{eq:generic-core-parameters}
\end{equation}
Substituting $u+nw=a_{\mathrm L}$, $v+mw=a_{\mathrm R}$, and
$w=-\gamma_+^2$ into \eqref{eq:universal-Jack-coefficient},
we define for general charges
\begin{equation*}
 \mathcal C_{\lambda,\mu}
 =\sum_{\substack{\nu\subseteq\lambda,\ \eta\subseteq\mu\\
                  |\nu|=|\eta|}}
 A_{\lambda,\nu}^{(\alpha)}(a_{\mathrm L})
 K_{\nu,\eta}^{(\alpha)}(-\gamma_+^2)
 A_{\mu,\eta}^{(\alpha)}(a_{\mathrm R}).
\end{equation*}

\begin{theorem}[General Jack formula for the rank-one irregular conformal block]
\label{thm:screening-Jack-evaluation}
For generic parameters as above, the rank-one irregular conformal
block \eqref{eq:rank-one-irregular-block}, with
$\Delta_\infty=\Delta(\gamma_\infty)$ and
$\Delta=\Delta(\gamma_z)$, satisfies
\begin{equation}
 \begin{aligned}
 \mathcal B(z)
 ={}&z^\kappa e^{\beta/z}
 \sum_{\lambda,\mu}z^{|\lambda|+|\mu|}\mathcal C_{\lambda,\mu}\\
 &\times (-\gamma_+c_1)^{-|\lambda|}
   \frac{(x_-)_\lambda^{(\theta)}(x_+)_\lambda^{(\theta)}}
        {H'_\theta(\lambda)}
 \ (\gamma_+c_1)^{-|\mu|}
   \frac{(y_-)_\mu^{(\theta)}(y_+)_\mu^{(\theta)}}
        {H'_\theta(\mu)}.
 \end{aligned}
 \label{eq:screening-evaluated-Jack-expansion}
\end{equation}
Both sums range over all partitions, with no length restrictions.
At each total degree the sum is finite.
\end{theorem}

\begin{proof}
First impose the screening specialization
\begin{equation*}
 p_{\mathrm t}=p+\gamma_z+m\gamma_+,\qquad
 2\rho-p_{\mathrm t}-\gamma_\infty=n\gamma_+,
 \qquad m,n\in\mathbb Z_{\ge0}.
\end{equation*}
Since $\gamma_+(\gamma_+-2\rho)=2$, the content parameters become
\[
 x_-=m\theta,\quad x_+=1+\gamma_+\gamma_z+(m-1)\theta,\qquad
 y_-=n\theta,\quad y_+=1+\gamma_+\gamma_\infty+(n-1)\theta.
\]
Moreover, $a_{\mathrm L}=u+nw$ and $a_{\mathrm R}=v+mw$ under
\eqref{eq:screening-universal-specialization}. Hence
\[
 \mathcal C_{\lambda,\mu}
 =C_{\lambda,\mu}^{(\alpha)}(u,v,w;m,n).
\]
The factors $(m\theta)_\lambda^{(\theta)}$ and
$(n\theta)_\mu^{(\theta)}$ vanish for
$\ell(\lambda)>m$ and $\ell(\mu)>n$, respectively.
For each $j\geq0$, substitute the coefficient of $z^j$ in
\eqref{eq:screening-Jack-expansion} into the integral formula
\eqref{eq:screening-block-coefficient-formula} of
Corollary~\ref{cor:screening-block-identification}.
Evaluating the resulting normalized Jack integrals by
\eqref{eq:screening-x-Jack-average} and
\eqref{eq:screening-y-Jack-average} proves
\eqref{eq:screening-evaluated-Jack-expansion} on the screening specialization.

To extend the identity, fix $j\geq0$ and consider the coefficient $b_j$
of $z^j$ in $z^{-\kappa}e^{-\beta/z}\mathcal B(z)$.
The polynomial PBW assertion of Theorem~\ref{thm:irregular-vertex-existence},
together with the highest-weight pairing, shows that $b_j$ is polynomial
in the irregular parameters, the external weights, and $\Lambda_2^{-1}$. Since $\Lambda_2=c_1^2/2$ and
$2\rho=\gamma_+-2/\gamma_+$, it is polynomial in
$p,p_{\mathrm t},\gamma_z,\gamma_\infty$ over
$\mathbb Q(\gamma_+,c_1)$. The finite Jack formula on the right has
the same polynomial dependence.

Introduce the independent affine coordinates
\[
 s=\frac{p_{\mathrm t}-p-\gamma_z}{\gamma_+},\qquad
 t=\frac{2\rho-p_{\mathrm t}-\gamma_\infty}{\gamma_+}
\]
in place of $p,p_{\mathrm t}$. For every $m,n\in\mathbb Z_{\ge0}$, the screening calculation shows that
the difference of the two coefficients at $s=m,t=n$ vanishes on a
nonempty convergence domain in $(\gamma_z,\gamma_\infty,\gamma_+,c_1)$.
Hence the difference at $s=m,t=n$ is the zero polynomial in
$\gamma_z,\gamma_\infty$ over $\mathbb Q(\gamma_+,c_1)$.
This proves \eqref{eq:screening-evaluated-Jack-expansion} coefficientwise.
\end{proof}

\begin{corollary}[Combinatorial formula at $c=1$]
\label{cor:c-one-combinatorial-formula}
Conjecture~\ref{conj:c-one-combinatorial-formula} holds. Explicitly, with
$S_{\lambda,\mu}$, $U_{\lambda/\nu}$, and $V_{\mu/\eta}$ as in
Section~\ref{subsec:combinatorial-conjecture}, one has
\begin{equation*}
 \begin{aligned}
 &\mathcal B(z;1,\theta_0^2,\theta_t^2,\theta_*,\tfrac14,\beta)\\
 &\quad=z^{-2\theta_t^2-2\beta(\theta_*-\beta)}e^{\beta/z}
  \sum_{\lambda,\mu\in\mathbb Y}z^{|\lambda|+|\mu|}
  S_{\lambda,\mu}
  \sum_{\substack{\nu\subseteq\lambda,\ \eta\subseteq\mu\\
                   |\nu|=|\eta|}}
       (-1)^{|\nu|}c_{\lambda,\mu}^{\nu,\eta}
       U_{\lambda/\nu}V_{\mu/\eta},
 \end{aligned}
\end{equation*}
where $c_{\lambda,\mu}^{\nu,\eta}\in\mathbb Z_{\ge0}$ are independent of
$\theta_0,\theta_t,\theta_*,\beta$.
\end{corollary}

\begin{proof}
At $c=1$, one has $\rho=0$ and $\alpha=\theta=1$.
Specializing Theorem~\ref{thm:screening-Jack-evaluation} and using the
parameters of \eqref{eq:c-one-parameters} gives the displayed expansion.
In terms of the standard tableau numbers $f^{\lambda/\nu}$ and
Littlewood--Richardson coefficients $\mathrm{LR}_{a,b}^{\nu}$,
we obtain, for $\nu\subseteq\lambda$, $\eta\subseteq\mu$, and $|\nu|=|\eta|=r$,
\begin{equation}
 c_{\lambda,\mu}^{\nu,\eta}
 =\frac{|\lambda|!\,|\mu|!\,f^{\lambda/\nu}f^{\mu/\eta}}
 {(|\lambda|-r)!(|\mu|-r)!\,f^\lambda f^\mu}
 \sum_{\substack{a,b\in\mathbb Y\\|a|+|b|=r}}
 \mathrm{LR}_{a,b}^{\nu}\mathrm{LR}_{a',b'}^{\eta}.
 \label{eq:c-one-explicit-coefficients}
\end{equation}
This follows from the shifted Schur identities
\cite[Theorems~8.1 and~11.1]{OkounkovOlshanski1998} and the LR
expansion of the squared dual Cauchy kernel
\cite[Chapter~I, (4.3$'$), (5.2), and~(9.1)]{Macdonald1995}.
The prefactor in \eqref{eq:c-one-explicit-coefficients} equals
$P_\nu^*(\lambda;1)P_\eta^*(\mu;1)$.
These shifted Schur values and the LR coefficients are nonnegative integers, hence
$c_{\lambda,\mu}^{\nu,\eta}\in\mathbb Z_{\ge0}$.
\end{proof}

In terms of symmetric-group characters, \eqref{eq:K-definition} at
$\alpha=1$ gives
\begin{equation*}
 \sum_{\substack{a,b\in\mathbb Y\\|a|+|b|=r}}
 \mathrm{LR}_{a,b}^{\nu}\mathrm{LR}_{a',b'}^{\eta}
 =\sum_{\varrho\vdash r}\frac{2^{\ell(\varrho)}}{z_\varrho}
   \chi^\nu(\varrho)\chi^{\eta'}(\varrho).
\end{equation*}
For $\varrho=(1^{m_1}2^{m_2}\cdots)\vdash r$ and $L\ge\ell(\nu)$,
with $\nu$ padded with zeros, the Frobenius alternant formula
\cite[Chapter~I, Section~7]{Macdonald1995} yields
\begin{equation*}
 \chi^\nu(\varrho)
 =\sum_{\pi\in S_L}\operatorname{sgn}(\pi)
 \sum_{\substack{
 n_{ik}\in\mathbb Z_{\ge0}\\
 \sum_{i=1}^{L}n_{ik}=m_k\ (1\le k\le r)\\
 \sum_{k=1}^{r}k n_{ik}=\nu_i-i+\pi(i)\ (1\le i\le L)}}
 \prod_{k=1}^{r}\frac{m_k!}{\prod_{i=1}^{L}n_{ik}!}.
\end{equation*}
\section{Hilbert schemes and geometric realization}
\label{sec:hilbert}

We realize the normalized block as a coefficientwise equivariant pairing
on Hilbert schemes of points on $\mathbb C^2$.
Incidence correspondences and Hopf operations define the operators in
this pairing, while tautological Euler classes define its two states.
The fixed-point expansion of the pairing gives the combinatorial formula
of Theorem~\ref{thm:screening-Jack-evaluation}.

\subsection{Hilbert--Fock space and geometric operators}
\label{subsec:Hilbert-Fock-space}

Let $T=(\mathbb C^\times)^2$ act on points of $\mathbb C^2$ by
$(t_1,t_2)(z,\xi)=(t_1z,t_2\xi)$, with additive weights
$\varepsilon_1,\varepsilon_2$. Set
\begin{equation*}
 R_T=\mathbb Q[\varepsilon_1,\varepsilon_2],\quad
 F_T=\mathbb Q(\varepsilon_1,\varepsilon_2),\quad
 X_N=\operatorname{Hilb}^N(\mathbb C^2).
\end{equation*}
The localized equivariant Borel--Moore groups, their direct sum, and its
completion with respect to the grading by $N$ are
\begin{equation*}
 \begin{split}
 \mathbb H_{T,N}&=H_\bullet^{T,\mathrm{BM}}(X_N)\otimes_{R_T}F_T,\\
 \mathbb H_T&=\bigoplus_{N\geq0}\mathbb H_{T,N},\qquad
 \widehat{\mathbb H}_T=\prod_{N\geq0}\mathbb H_{T,N}.
 \end{split}
\end{equation*}
We use localized Poincar\'e duality to pass to cohomology when taking
restrictions or multiplying classes; see
\cite[Sections~1(ii) and~2(i)]{NakajimaMoreLectures}.
Fix the parameter dictionary
\begin{equation*}
 \alpha=-\frac{\varepsilon_2}{\varepsilon_1},\quad
 \theta=-\frac{\varepsilon_1}{\varepsilon_2},\quad
 \hbar=-\varepsilon_2,\quad
 (\varepsilon_1,\varepsilon_2)=(\theta\hbar,-\hbar).
\end{equation*}
Let $\mathbb K=F_T(\mathbf v)$, where $\mathbf v$ is the finite list
of independent scalar cap and core parameters in use, excluding the
expansion variable $z$.
Extend coefficients before completing:
\begin{equation*}
 \mathbb H_{\mathbb K,N}=\mathbb H_{T,N}\otimes_{F_T}\mathbb K,
 \quad \mathbb H_{\mathbb K}=\bigoplus_N\mathbb H_{\mathbb K,N},
 \quad \widehat{\mathbb H}_{\mathbb K}=\prod_N\mathbb H_{\mathbb K,N}.
\end{equation*}
With Nakajima's labelling \cite[equation~(2.10)]{NakajimaMoreLectures},
$I_\lambda=(\xi^{\lambda_1},z\xi^{\lambda_2},\ldots)$ and the quotient
monomial for $s=(i,j)$ is $z^{i-1}\xi^{j-1}$.
Coordinate functions have the contragredient weights
$-\varepsilon_1,-\varepsilon_2$.
Write $[I_\lambda]$ for the pushed-forward fixed-point fundamental class.
For $u,v\in\mathbb H_{T,N}$ define
\begin{equation*}
 \langle u,v\rangle_T=(-1)^N\sum_{\lambda\vdash N}
 \frac{u|_{I_\lambda}v|_{I_\lambda}}{e_T(T_{I_\lambda}X_N)}.
\end{equation*}
Different degrees are orthogonal, and the pairing extends to
$\mathbb H_T\times\widehat{\mathbb H}_T$.

By Nakajima's fixed-point description, the tangent weights at $I_\lambda$ are
\[
(\ell(s)+1)\varepsilon_1-a(s)\varepsilon_2,\qquad
(a(s)+1)\varepsilon_2-\ell(s)\varepsilon_1,
\qquad s\in\lambda.
\]
With respect to the one-parameter subgroup $t\mapsto(t^{-1},t)$,
the first family of tangent weights is negative and the second is positive.
Their negative and positive Euler factors are therefore
\[
e_-(\lambda)=\hbar^{|\lambda|}H'_\theta(\lambda),\qquad
e_+(\lambda)=(-\hbar)^{|\lambda|}H_\theta(\lambda).
\]
The Hilbert--Jack identification sends the monic Jack polynomial
$P_\lambda^{(\alpha)}$ to
\begin{equation}\label{eq:Hilbert-Jack-identification}
  P_\lambda:=\frac{[I_\lambda]}{e_-(\lambda)},
\end{equation}
and identifies the pairing with the Jack scalar product:
\begin{equation}
\label{eq:fixed-point-Jack-pairing} 
\langle P_\lambda,P_\mu\rangle_T
=
\frac{\delta_{\lambda\mu}}{b_\lambda(\alpha)}.
\end{equation}

Let $\mathcal Z_N\subset X_N\times\mathbb C^2$ be the universal
subscheme, with finite flat projection $\pi_X$ to $X_N$. The
rank-$N$ tautological bundle is
\begin{equation*}
 \mathcal V_N=(\pi_X)_*\mathcal O_{\mathcal Z_N},\qquad
 (\mathcal V_N)_{I_\lambda}=\mathbb C[z,\xi]/I_\lambda.
\end{equation*}
First introduce an independent flavor torus $T_u$ with weight $m_u$ and form
\begin{equation*}
 \widetilde{\mathcal E}_{m_u}^{(N)}
 =\hbar^{-N}e_{T\times T_u}(\mathcal V_N\otimes\mathbb C_{m_u}).
\end{equation*}
Then apply the coefficient-ring map $F_T[m_u]\to F_T[u]$,
$m_u\mapsto\hbar u$, and define
\begin{equation*}
 \mathcal E_u^{(N)}=
 \left.\widetilde{\mathcal E}_{m_u}^{(N)}\right|_{m_u=\hbar u},
 \qquad \mathsf M_u=\bigoplus_N(\mathcal E_u^{(N)}\cap-).
\end{equation*}
Hereafter expressions involving $\mathbb C_{\hbar u}$ are understood
via this base change.

\begin{proposition}[Tautological content multiplier]
\label{prop:content-multiplier}
The normalized Euler-class action is self-adjoint and satisfies
\begin{equation*}
 \mathsf M_uP_\lambda=(u)_\lambda^{(\theta)}P_\lambda.
\end{equation*}
\end{proposition}
\begin{proof}
The quotient monomial of a box has weight
$-(i-1)\varepsilon_1-(j-1)\varepsilon_2
=\hbar\operatorname{ct}_\theta(s)$.
The restriction of $\widetilde{\mathcal E}_{m_u}^{(N)}$ to $I_\lambda$
is therefore
$\prod_{s\in\lambda}(m_u/\hbar+\operatorname{ct}_\theta(s))$.
Base change proves the eigenvalue formula; diagonality and
\eqref{eq:fixed-point-Jack-pairing} prove self-adjointness.
\end{proof}
The Nakajima correspondence
$Q^{[N+r,N]}\subset X_{N+r}\times\mathbb C^2\times X_N$ consists of
$(J,x,I)$ with $J\subset I$ and $I/J$ of length $r$ supported at $x$.
With projections $p_+,\varrho,p_-$, its creation operator
$P_{-r}(\gamma)$, for $\gamma\in H_T^*(\mathbb C^2)$, is defined by
pulling back the input and $\gamma$ along $p_-$ and $\varrho$,
capping with $[Q^{[N+r,N]}]$, and pushing forward along $p_+$.
By \cite[equation~(3.9), Proposition~3.10]{NakajimaMoreLectures},
\begin{equation}
 \mathfrak a_{-r}=P_{-r}(\varepsilon_2)\longleftrightarrow p_r,
 \qquad \mathfrak a_r:=\mathfrak a_{-r}^{\dagger}
 \longleftrightarrow\alpha r\partial_{p_r}\quad(r>0).
 \label{eq:Nakajima-mode-normalization}
\end{equation}
Thus $[\mathfrak a_r,\mathfrak a_{-s}]=\alpha r\delta_{rs}$.
For $r=1$, the nested Hilbert scheme is the blowup of
$X_N\times\mathbb C^2$ along $\mathcal Z_N$, with exceptional divisor
$E^{[N+1,N]}$ (empty for $N=0$).

\begin{proposition}
\label{prop:Nakajima-flow}
The operator
\begin{equation*}
 \mathsf D_a=\alpha^{-1}\sum_{r\geq1}
 \mathfrak a_{-(r+1)}\mathfrak a_r+a\mathfrak a_{-1}
\end{equation*}
is represented from degree $N$ to degree $N+1$ by the correspondence class
\begin{equation}
 [E^{[N+1,N]}]+a\,\varrho^*(\varepsilon_2)\cap[Q^{[N+1,N]}].
 \label{eq:finite-Nakajima-correspondence}
\end{equation}
The correspondence \eqref{eq:finite-Nakajima-correspondence}, together with
the degree operator $\mathsf E|_{\mathbb H_{\mathbb K,N}}=N\,\mathrm{id}$, realizes
\[
 \mathsf T_a(z)=e^{z\mathsf D_a}z^{\mathsf E},
\]
where the powers of $\mathsf D_a$ in the exponential are formed by convolution.
\end{proposition}
\begin{proof}
Define $\partial$ degreewise by
\[
 \partial|_{\mathbb H_{\mathbb K,N}}=c_1(\mathcal V_N)\cap-.
\]
With the normalization \eqref{eq:Nakajima-mode-normalization},
Lehn's quotient-line formula
\cite[Section~1.2, equation~(5)]{Lehn1999} and
Nakajima's equivariant commutator formula
\cite[Theorem~4.5]{NakajimaMoreLectures} identify the correspondence
operator defined by $[E^{[N+1,N]}]$ with the restriction to
$\mathbb H_{\mathbb K,N}$ of
\[
 -[\partial,P_{-1}(1)]
 =\alpha^{-1}\sum_{r\ge1}\mathfrak a_{-(r+1)}\mathfrak a_r.
\]
Adding $aP_{-1}(\varepsilon_2)=a\mathfrak a_{-1}$ gives
\eqref{eq:finite-Nakajima-correspondence}.
\end{proof}

The closures of the disjoint-union graphs and their transposes define
the union product and splitting coproduct
\[
 m:\mathbb H_T\otimes_{F_T}\mathbb H_T\longrightarrow\mathbb H_T,
 \qquad
 \Delta:\mathbb H_T\longrightarrow
 \mathbb H_T\otimes_{F_T}\mathbb H_T.
\]
These equip $\mathbb H_T$ with a connected, commutative and
cocommutative Hopf algebra structure, with unit given by the vacuum
and counit by projection onto $\mathbb H_{T,0}$; see
\cite[\S9.2, Theorem~9.10(1)]{Nakajima1999}.
Under \eqref{eq:Hilbert-Jack-identification}, $m$ becomes the usual
product of symmetric functions
\cite[Theorem~9.10(2) and equation~(9.13)]{Nakajima1999}, and $\Delta$
becomes the usual coproduct by adjointness. Thus the power sums $p_r$
are primitive.

The geometric form of \eqref{eq:K-operator-definition} is
\begin{equation}
 N_{\mathrm{osc}}
 =\alpha^{-1}\sum_{r\ge1}
   \frac{\mathfrak a_{-r}\mathfrak a_r}{r},
 \qquad
 \mathsf K(w)=(\alpha w)^{N_{\mathrm{osc}}}.
 \label{eq:geometric-K}
\end{equation}
The Hopf-power formula for symmetric functions
\cite[Theorem~3.13 and \S4.3]{DiaconisPangRam2014}
gives, for $d=\alpha w\in\mathbb Z_{\ge0}$,
\begin{equation*}
 \mathsf K(w)=m^{(d)}\circ\Delta^{(d)},
\end{equation*}
where $m^{(d)}$ and $\Delta^{(d)}$ denote the $d$-fold product and
coproduct, with the zero-fold Hopf power understood as
$\mathrm{unit}\circ\mathrm{counit}$.
For general $d$, $\mathsf K(w)$ is the polynomial continuation of
these Hopf powers on each $\mathbb H_{T,N}$.

Let $\mathsf S$ be the Hopf antipode. Since the $p_r$ are primitive
and $\alpha\gamma_+^2=2$ by \eqref{eq:screening-Jack-parameter},
\[
 \mathsf K(-\gamma_+^2)=(-2)^{N_{\mathrm{osc}}}
 =\mathsf S\circ m\circ\Delta.
\]

\subsection{Laguerre caps and geometric realization}
\label{subsec:Laguerre-caps}

\begin{definition}[Laguerre cap]
Let $M,s$ be generic scalar parameters in $\mathbb K$ and let $q$ be
invertible. Define the completed cap state by
\begin{equation*}
 \ket{\mathscr L_{M;s;q}}=\mathsf M_s
 \exp\!\left(M\theta\sum_{r\ge1}\frac{p_r}{rq^r}\right)\vac.
\end{equation*}
Its functional is
$\mathscr L_{M;s;q}(v)=\langle v,\mathscr L_{M;s;q}\rangle_T$.
\end{definition}

\begin{proposition}[Laguerre cap]
\label{prop:Laguerre-cap}
The vector and functional coefficients of the Laguerre cap are
\begin{align}
 [P_\lambda]\ket{\mathscr L_{M;s;q}}
 &=q^{-|\lambda|}\frac{(M\theta)_\lambda^{(\theta)}
 (s)_\lambda^{(\theta)}}{H_\theta(\lambda)},
 \label{eq:Laguerre-vector-coefficient}\\
 \mathscr L_{M;s;q}(P_\lambda)
 &=q^{-|\lambda|}\frac{(M\theta)_\lambda^{(\theta)}
 (s)_\lambda^{(\theta)}}{H'_\theta(\lambda)}.
 \label{eq:Laguerre-functional-coefficient}
\end{align}
The degree-$N$ component of $\ket{\mathscr L_{M;s;q}}$ is the localized
Borel--Moore homology class
\begin{equation}
 (-\hbar q)^{-N}e_T(\mathcal V_N\otimes\mathbb C_{\hbar M\theta})
 e_T(\mathcal V_N\otimes\mathbb C_{\hbar s})\cap[X_N].
 \label{eq:Laguerre-geometric-class}
\end{equation}
\end{proposition}

\begin{proof}
The Jack Cauchy identity
\cite[Chapter~VI, Section~10]{Macdonald1995} and
Proposition~\ref{prop:content-multiplier} give
\eqref{eq:Laguerre-vector-coefficient}.
Pairing with $P_\lambda$ and using
\eqref{eq:fixed-point-Jack-pairing} gives
\eqref{eq:Laguerre-functional-coefficient}.
Localization gives $[P_\lambda](f\cap[X_N])=f|_{I_\lambda}/e_+(\lambda)$.
Applying this to \eqref{eq:Laguerre-geometric-class} reproduces
\eqref{eq:Laguerre-vector-coefficient}.
\end{proof}

For $M\in\mathbb Z_{\ge0}$ and $s=a+1+(M-1)\theta$,
\eqref{eq:Laguerre-functional-coefficient} is the normalized
Laguerre--Jack average used in \eqref{eq:screening-x-Jack-average}
and \eqref{eq:screening-y-Jack-average}.

Use the prefactor, cap, and core parameters defined in
\eqref{eq:generic-prefactor-parameters}--\eqref{eq:generic-core-parameters}.
Set
\begin{equation}
 \mathscr L_x=\mathscr L_{x_-/\theta\,;\,x_+\,;\,-\gamma_+c_1},
 \qquad
 \mathscr L_y=\mathscr L_{y_-/\theta\,;\,y_+\,;\,\gamma_+c_1}.
 \label{eq:generic-Laguerre-caps}
\end{equation}
\begin{theorem}
\label{thm:rank-one-Hilbert-Fock}
With the parameters of Theorem~\ref{thm:screening-Jack-evaluation},
the rank-one irregular conformal block satisfies
\begin{equation}
 \label{eq:rank-one-Hilbert-Fock-matrix-element}
 \mathcal B(z)
 =z^\kappa e^{\beta/z}
 \left\langle
 \mathscr L_x\middle|
 \mathsf T_{a_{\mathrm L}}(z)\mathsf K(-\gamma_+^2)
 \mathsf T_{a_{\mathrm R}}(z)^\dagger
 \middle|\mathscr L_y
 \right\rangle_T.
\end{equation}
The matrix element in \eqref{eq:rank-one-Hilbert-Fock-matrix-element}
is a formal scalar series in $\mathbb K[[z]]$.
Its coefficient of $z^d$ is a finite sum over pairs $(\lambda,\mu)$
with $|\lambda|+|\mu|=d$.
\end{theorem}

\begin{proof}
By \eqref{eq:Nakajima-mode-normalization}, Proposition~\ref{prop:Nakajima-flow},
and \eqref{eq:geometric-K}, the Hilbert--Jack identification
\eqref{eq:Hilbert-Jack-identification} sends the geometric operators
$\mathsf T_a(z)$ and $\mathsf K(w)$ to those of
Section~\ref{subsec:Fock-flow}. It preserves the pairing
\eqref{eq:fixed-point-Jack-pairing}, hence also the adjoints.
The matrix formula \eqref{eq:universal-sandwich-matrix} therefore gives
\[
 \left\langle\mathscr L_x\middle|
 \mathsf T_{a_{\mathrm L}}(z)\mathsf K(-\gamma_+^2)
 \mathsf T_{a_{\mathrm R}}(z)^\dagger
 \middle|\mathscr L_y\right\rangle_T
 =\sum_{\lambda,\mu}z^{|\lambda|+|\mu|}
 \mathcal C_{\lambda,\mu}\,
 \mathscr L_x(P_\lambda)\mathscr L_y(P_\mu).
\]
At each power $z^d$, only pairs with $|\lambda|+|\mu|=d$ contribute,
so the contraction is finite. Substituting the cap evaluation
\eqref{eq:Laguerre-functional-coefficient}, with the parameters specified in
\eqref{eq:generic-Laguerre-caps}, identifies the sum with
$z^{-\kappa}e^{-\beta/z}\mathcal B(z)$ by
\eqref{eq:screening-evaluated-Jack-expansion}.
\end{proof}

\paragraph*{Acknowledgement.}
This work was supported by JSPS KAKENHI Grant Number 22K03350.

\end{document}